\documentclass[11pt]{article}

\usepackage
{
makeidx,epsfig, amstext, setspace, amsmath, amssymb,amsthm,amsfonts,latexsym,algorithm,algorithmic,
wrapfig,bbm,graphicx,color,enumerate,endnotes, tabularx, enumitem, xspace, todonotes, comment
}
\usepackage{tabularx}
\usepackage{thmtools}
\usepackage{thm-restate}

\usepackage[letterpaper,top=2.5cm, bottom= 2.5cm, left= 2.5cm, right= 2.5cm]{geometry}

\newtheorem{theorem}{Theorem}[section]

\newcommand{\relemma}{lemma}
\newcommand{\retheorem}{theorem}
\newcommand{\restatelemma}[1]{\renewcommand{\relemma}{rlemma}#1\renewcommand{\relemma}{lemma}}
\newcommand{\restatetheorem}[1]{\renewcommand{\retheorem}{rtheorem}#1\renewcommand{\retheorem}{theorem}}

\newtheorem{OB}{\bf Observation}

\newcounter{claim_nb}[theorem]
\newtheorem*{claim*}{Claim}

\newtheorem{lemma}[theorem]{Lemma}
\newtheorem{corollary}[theorem]{Corollary}

\theoremstyle{definition}
\newtheorem{definition}[theorem]{Definition}

\newcommand{\zC}{\mathcal C}

\newcommand{\zF}{\mathcal F}

\newcommand{\zP}{\mathcal P}

\newcommand{\ignore}[1]{}

\makeatletter
\def\th@plain{%
  \thm@notefont{}
  \itshape 
}
\def\th@definition{%
  \thm@notefont{}
  \normalfont 
}
\makeatother

\setlist{noitemsep}

\newcommand{\T}{\mathcal{T}}
\newcommand{\cH}{\mathcal{H}}

\newcommand{\executeiffilenewer}[3]{%
\ifnum\pdfstrcmp{\pdffilemoddate{#1}}%
{\pdffilemoddate{#2}}>0%
{\immediate\write18{#3}}\fi%
} 

\newcommand{\svg}[2]{\def\svgwidth{#1}%
\executeiffilenewer{#2.svg}{#2.pdf}%
{inkscape -z -D --file=#2.svg %
--export-pdf=#2.pdf --export-latex}%
{\input{#2.pdf_tex}}}

\newcommand{\R}{\mathcal{R}}
\newcommand{\bR}{\overline{\mathcal{R}}}

\newcommand{\Rmvec}{R^{\textup{vec}}_{d,r}}
\newcommand{\Rm}{R_{z,r}}
\newcommand{\Rmx}{R_{z,r^*}}

\newcommand{\sm}{I_{k,r,z}}
\newcommand{\smsep}{I^{\textup{sep}}_{k,r,z}}

\newcommand{\classWone}{\textup{W[1]}}
\newcommand{\FPT}{\textup{FPT}\xspace}
\newcommand{\EP}{Erd\H os-P\'osa\xspace}
\newcommand{\defparproblemu}[4]{
 
\begin{center}
\noindent\fbox{
  \begin{minipage}{0.9\textwidth}
  \begin{tabular*}{\textwidth}{@{\extracolsep{\fill}}lr} #1 &
 \\ \end{tabular*}
  {\bf{Input:}} #2  \\
  {\bf{Find:}} #4
  \end{minipage}
  }
\end{center}
}

\newcommand{\gridtiling}{{\sc{Grid Tiling}}\xspace} 
\newcommand{\skewdisjoint}{{\sc{Maximum Skew Disjoint Paths}}\xspace} 
\newcommand{\maxdisjoint}{{\sc{Maximum Disjoint Paths}}\xspace} 
\newcommand{\disjointpath}{{\sc{Disjoint Paths}}\xspace} 
\newcommand{\skewdisjointx}{{\sc{Maximum Skew Disjoint Paths$^*$}}\xspace} 
 
\newcommand{\maxhhdisjoint}{{\sc{Maximum Disjoint $\cH$-Paths}}\xspace}

\newif\ifabstract
\abstracttrue

\abstractfalse 

\newif\iffull
\ifabstract 
\usepackage[nomarkers,nolists]{endfloat}
\usepackage{mathptmx}
\DeclareMathAlphabet{\mathcal}{OMS}{cmsy}{m}{n}

\fullfalse \else \fulltrue  \fi

\date{}
\author{
D\'aniel Marx\thanks{
	Institute for Computer Science and Control, Hungarian Academy of Sciences (MTA SZTAKI)
	\texttt{dmarx@cs.bme.hu}. Research  supported by the European Research Council (ERC)  grant 
``PARAMTIGHT: Parameterized complexity and the search for tight
complexity results,'' reference 280152 and OTKA grant NK105645.
  }
\and
Paul Wollan\thanks{Department of Computer Science, University of Rome, \texttt{wollan@di.uniroma1.it}.  Research supported by the European Research Council under the European Union's Seventh Framework Programme (FP7/2007-2013)/ERC Grant Agreement no. 279558.}
}
\begin{document}

\ifabstract
\begin{titlepage}
\fi
\title{An exact characterization of tractable demand patterns for maximum disjoint path problems}
\maketitle
%
\begin{abstract}
  We study the following general disjoint paths problem: given a
  supply graph $G$, a set $T\subseteq V(G)$ of terminals, a demand
  graph $H$ on the vertices $T$, and an integer $k$, the task is to
  find a set of $k$ pairwise vertex-disjoint valid paths, where we say
  that a path of the supply graph $G$ is valid if its endpoints are in
  $T$ and adjacent in the demand graph $H$.  For a class $\cH$ of
  graphs, we denote by \maxhhdisjoint the restriction of this problem
  when the demand graph $H$ is assumed to be a member of $\cH$.  We
  study the fixed-parameter tractability of this family of problems,
  parameterized by $k$.  Our main result is a complete
  characterization of the fixed-parameter tractable cases of
  \maxhhdisjoint for every hereditary class $\cH$ of graphs: it turns
  out that complexity depends on the existence of large induced
  matchings and large induced skew bicliques in the demand graph $H$
  (a skew biclique is a bipartite graph on vertices $a_1$, $\dots$,
  $a_n$, $b_1$, $\dots$, $b_n$ with $a_i$ and $b_j$ being adjacent if
  and only if $i\le j$).  Specifically, we prove the following
  classification for every hereditary class $\cH$.
\begin{itemize}
\item If $\cH$ does not contain every matching and does not contain every skew biclique, then \maxhhdisjoint is FPT.
\item If $\cH$ does not contain every matching, but contains every skew biclique, then \maxhhdisjoint is W[1]-hard, admits an FPT approximation, and the valid paths satisfy an analog of the \EP property.
\item If $\cH$ contains every matching, then \maxhhdisjoint is W[1]-hard and the valid paths do not satisfy the analog of the \EP property.
\end{itemize}

\end{abstract}

\ifabstract
\end{titlepage}
\fi
%

\clearpage

\iffull
\tableofcontents
\clearpage
\fi
\section{Introduction}
\label{sec:ntro}

Given an undirected graph $G$ and pairs of vertices $(s_1,t_1)$,
$\dots$, $(s_k,t_k)$, the \disjointpath problem asks for pairwise
vertex-disjoint paths $P_1$, $\dots$, $P_k$ such that $P_i$ has
endpoints $s_i$ and $t_i$. A celebrated result of Robertson and
Seymour \cite{MR97b:05088} (see also \cite{Kawarabayashi:2010:SPG:1806689.1806784}) states that \disjointpath can be solved in
time $f(k)n^3$ for some function $f$ depending only on $k$, that is,
there is a cubic-time algorithm for every fixed $k$. Therefore,
\disjointpath is not only polynomial-time solvable for every fixed
$k$, but fixed-parameter tractable parameterized by $k$. Recall that a
problem is {\em fixed-parameter tractable (\FPT)} parameterized by $k$
if it can be solved in time $f(k)n^{O(1)}$ for some computable
function $f$ depending only on $k$.  
\begin{theorem}[Robertson and Seymour \cite{MR97b:05088}]\label{thm:kdisjointpath}
\disjointpath can be solved in time $f(k)\cdot n^{O(1)}$.
\end{theorem}

The main focus of the present paper is a natural maximization version
of \disjointpath. Given an undirected graph $G$, pairs of vertices
$(s_1,t_1)$, $\dots$, $(s_m,t_m)$, and an integer $k$, the
\maxdisjoint problem asks for a set of $k$ pairwise vertex-disjoint {\em
  valid} paths, where we say that a path is valid if its endpoints are
$s_j$ and $t_j$ for some $1\le j \le m$.   We will typically refer to the graph $G$ as the \emph{supply graph} and the graph with vertex set $\{s_1, \dots, s_m, t_1, \dots, t_m\}$ and edge set $s_it_i$ for $1 \le i \le m$ as the \emph{demand graph}.  The \maxdisjoint problem remains NP-complete even with strong restrictions on the input:  it is NP-complete when restricted to problem instances with supply graph $G$ and demand graph $H$ such that $G \cup H$ is planar \cite{MP93}.  See \cite{NS09} for an in depth discussion of variants of the problem that are known to be computationally hard, as well as \cite{Schrijver:1990:PFV:574821} for surveys on the problem.

In contrast, for every fixed $k$ it is easy to see that \maxdisjoint
is polynomial-time solvable: we guess $k$ integers $1\le j_1 <j_2
<\dots <j_k \le m$, and then solve the \disjointpath instance on $G$
with pairs $(s_{j_1},t_{j_2})$, $\dots$, $(s_{j_k},t_{j_k})$ using the
algorithm of Theorem~\ref{thm:kdisjointpath}. Clearly, the
\maxdisjoint instance has a solution if and only if at least one of
the instances of \disjointpath has. As there are $m^{O(k)}$ different
ways of selecting the $k$ integers $j_1$, $\dots$, $j_k$, this results
in an $f(k)n^{O(k)}$ time algorithm. But is \maxdisjoint
fixed-parameter tractable? As we shall see later in this paper,
\maxdisjoint is \classWone-hard, which means that it is not \FPT under
standard complexity assumptions. The hardness result holds even if $G$
is a planar graph whose treewidth is bounded by a function of
$k$. This indicates that two fundamental algorithmic ideas underlying
the \disjointpath algorithm of Robertson and Seymour
\cite{MR97b:05088} cannot be used for \maxdisjoint: finding irrelevant
vertices exploiting properties of graphs embedded on surfaces (or
excluding minors) and using dynamic programming to solve
bounded-treewidth instances.

Despite the hardness of the general problem, there are easier
special cases of \maxdisjoint: classic results yield polynomial time 
algorithms even with $k$ as part of the input when the problem is
 restricted to certain types of demand graphs. Suppose that $S$ and 
 $T$ are two sets
of vertices and the set of pairs given in the input is $S\times T$
(that is, every pair $(s,t)$ with $s\in S$, $t\in T$ is listed in the
input; note that the problem definition does not require the pairs
to be disjoint). Then the valid paths are the paths connecting $S$ and
$T$, hence it can be checked in polynomial time if there are $k$ valid
paths by solving a maximum flow problem with vertex capacities.  
The demand graphs in this case are complete bipartite graphs.
 A result of Mader 
\cite{M78} generalizes this observation by giving a min-max theorem for the maximum number of 
disjoint valid paths when the demand graph is a multi-partite graph.  
Mader's theorem is existential, but a maximal set of disjoint valid paths 
can be algorithmically found in polynomial time as an application of 
Lov\'asz' matroid matching algorithm \cite{Lovasz1980208}.  In a recent paper, Hirai and Pap \cite{HP13} exactly characterized which demand graphs make a more general version of the \emph{weighted edge-disjoint paths} problem polynomial time solvable. 

It is possible to use the Robertson-Seymour algorithm for the \disjointpath problem to find instances that are \FPT parameterized by the number $k$ of paths, but NP-complete when $k$ is part of the input.  Consider for example the case when the set of pairs is
$(S_1\times T_1)\cup (S_2\times T_2)$ for pairwise disjoint subsets
$S_1,S_2,T_1,T_2\subseteq V(G)$. This case of the problem is a restatement of the node-capacitated 2-commodity flow problem and is NP-complete when $k$ is included in the input \cite{Even:1975:CTT:1382429.1382492}.  To show that this case is \FPT, we can proceed in  the following
way. First, we guess the number $0\le k_1 \le k$ of paths in the
solution that connect $S_1$ and $T_1$ (and hence $k_2=k-k_1$ paths
connect $S_2$ and $T_2$). Let us introduce $k_1$ vertices $s^1_1$,
$\dots$, $s^1_{k_1}$, all of them fully connected to $S_1$; another
$k_1$ vertices $t^1_1$, $\dots$, $t^1_{k_1}$, all of them fully
connected to $T_1$. Similarly, we introduce $k_2$ vertices $s^2_1$,
$\dots$, $s^2_{k_2}$ fully connected to $S_1$ and $k_2$ vertices
$t^2_1$, $\dots$, $t^2_{k_2}$ fully connected to $T_2$. Then the
required $k_1+k_2$ paths exist if the \disjointpath instance with
pairs $(s^1_1,t^1_1)$, $\dots$, $(s^1_{k_1},t^1_{k_1})$,
$(s^2_1,t^2_{1})$, $\dots$, $(s^2_{k_2},t^2_{k_2})$ has a
solution. Therefore, we can reduce the problem to $k+1$ instances of
\disjointpath, implying that this special case of \maxdisjoint is \FPT.

Our main goal is to understand which demand patterns make \maxdisjoint
fixed-parameter tractable. The formal setting of our investigations is
the following. First, we introduce a slightly different formulation of
\maxdisjoint. Let $G$ be the supply graph, $T\subseteq V(G)$ be a set
of terminals, and $H$ be the demand graph defined on the vertices
$T$. We say that a path in $G$ is {\em valid} if both of its
endpoints are in $T$ and they are adjacent in $H$. The task is now to
find $k$ pairwise vertex-disjoint valid paths. The examples above can be expressed by an instance where $H$ is a
biclique (complete bipartite graph) or the disjoint union of two bicliques. For a class
$\cH$ of graphs, we define \maxhhdisjoint as the special case \maxdisjoint when $H$ is restricted to be a member of $\cH$.
\iffull

\defparproblemu{\maxhhdisjoint}
{
A graph $G$, a subset $T$ of vertices, a graph $H\in \cH$ on $T$, an integer $k$.
}{}
{
A set of $k$ pairwise vertex-disjoint paths in $G$ such that each path connects some $x,y\in T$ that are adjacent in $H$.
}

\fi
For example, as we have seen, if $\cH$ is the class of all
bicliques, then \maxhhdisjoint is polynomial-time solvable and if
every graph in $\cH$ is the disjoint union of two bicliques, then
\maxhhdisjoint is fixed-parameter tractable. One can observe that the
argument can be generalized to the case when the two bicliques are not
disjoint (i.e., the demand graph $H$ graph is obtained by fully
connecting $S_1$ with $T_1$ and $S_2$ with $T_2$, where these four
sets are not necessarily disjoint), or to the case where every graph
in $\cH$ is the (not necessarily disjoint) union of $c$ bicliques for
some constant $c$, or to the case where every graph $H\in \cH$ has the
property that the vertices in $H$ have at most $c$ different
neighborhoods for some constant $c$. Therefore, there are fairly
complicated demand patterns that make the problem \FPT.

Formally, our goal is to identify every class $\cH$ for which
\maxhhdisjoint is \FPT. For technical reasons, we restrict our
attention to classes $\cH$ that are {\em hereditary,} that is, closed
under taking induced subgraphs. Intuitively, if $H'$ is an induced
subgraph of some $H\in \cH$, then adding $H'$ to $\cH$ should not make
the problem any harder: given an instance with demand pattern $H'$, we
can easily express it with demand pattern $H$ by introducing dummy
isolated terminals into the supply graph $G$ to represent the vertices
$V(H)\setminus V(H')$. Therefore, it seems justified to study only
graph classes that are closed under taking induced subgraphs. However,
there is no formal reduction showing that if every graph in $\cH'$ is
an induced subgraph of a member of $\cH$, then the fixed-parameter
tractability of the problem with $\cH$ implies the fixed-parameter
tractability of the problem with $\cH'$.  There are at least two
technical issues with the simple reduction described above: first, adding the
isolated vertices may increase the size of the instance if $H$ is much
larger than $H'$ and, second, even if we know that $H'\in \cH'$ is a subgraph
of some $H\in \cH$, finding such an $H$ may be computationally hard. Therefore, to
avoid the discussion of artificial technicalities, we consider only
hereditary classes.

\textbf{Our results.} First, we investigate a purely combinatorial
question.  A classical result of Erd\H os and P\'osa \cite{MR0175810}
states that in every undirected graph $G$, the minimum number of
vertices needed to cover every cycle in $G$ can be bounded by a
function of the maximum number of vertex-disjoint cycles. This result
motivates the following definition: we say that a set $\zC$ of graphs
has the {\em \EP property} if there is a function $f(k)$ such that
every graph $G$ has either $k$ vertex-disjoint subgraphs that belong
to $\zC$ or a set $X$ of at most $f(k)$ vertices such that $G-X$ has
no subgraph that belongs to $\zC$; the result of Erd\H os and P\'osa
\cite{MR0175810} can be stated as saying that the set of all cycles
has this property. The literature contains numerous results proving
that the \EP property holds for variants of the disjoint cycle problem
 such as disjoint long cycles \cite{BBR07}, directed cycles \cite{RRST95}, 
 cycles of length 0 mod $m$ \cite{T88}, as well 
 as characterizing when the \EP property holds for odd cycles \cite{R99, 
 T01, RR01, Kawarabayashi:2006:NDC:1294364.1294418} and cycles of non-zero length mod $m$ 
 \cite{W11}.  Further study has considered whether sets $\zC$ defined by other 
 containment relations such as minors also have the \EP property \cite{RS5, 
 DKW12}.
 
We investigate the natural analog of the \EP property in the context
of the \maxdisjoint problem: Is it true that the valid paths have the
\EP property, that is, is it true that either there are $k$ valid paths
or a set of at most $f(k)$ vertices covering every valid path? Besides
its combinatorial interest, we explore this question because the \EP
property of some objects is often correlated with good algorithmic
behavior of the corresponding packing/covering problems, especially
from the viewpoint of fixed-parameter tractability.
However, in general, the answer to this question is no. The standard
counterexample is an $n\times n$ grid graph with the vertices $s_1$,
$\dots$, $s_n$ appearing in the top row from left to right, and the
vertices $t_1$, $\dots$, $t_n$ appearing in the bottom row from right
to left. Then every $s_i-t_i$ path intersects every $s_j-t_j$ path for
$i\neq j$, but we need $n-1$ vertices to cover all such
paths. Therefore, the \EP property does not hold for valid paths in
general, but may hold for the \maxhhdisjoint problem for certain
(hereditary) classes $\cH$. For example, if $\cH$ contains only
bicliques, then Menger's Theorem states that the \EP property holds in
a tight way with $f(k)=k-1$; if $\cH$ contains only cliques, then a
classical result of Gallai \cite{MR0130188} states that the the \EP
property holds with $f(k)=2k-2$.

Let $M_r$ be the graph consisting of a
matching of size $r$ (i.e., $M_r$ has $2r$ vertices and $r$
edges). The counterexample above shows that if the hereditary class
$\cH$ contains $M_r$ for every $r\ge 1$, then the \EP property surely
does not hold. Surprisingly, this is the only obstacle: our first
result states that if $\cH$ is a hereditary class of graphs not
containing $M_r$ for every $r\ge 1$, then the valid paths in 
\maxhhdisjoint have the \EP property. Our proof is
algorithmic and gives an algorithm that either produces a set of
disjoint valid paths or a hitting set $Z$ covering every valid path.
\begin{theorem}[Excluding large induced matching implies \EP property]\label{thm:mainEP}
Let $\cH$ be a hereditary class of graphs, and assume there exists an integer $r \ge  1$ such that $M_r \notin \cH$.  There exists an algorithm which given a graph $G$, $T \subseteq V(G)$, integer $k\ge 1$, and $H\in \cH$  with $V(H) = T$, returns one of the following:
\begin{enumerate}
\item a set of $k$ pairwise disjoint valid paths, or
\item a set $Z$ of at most 
$2^{O(k+r)}$ vertices such that every valid path intersects $Z$.
\end{enumerate}
Moreover, the algorithm runs in time $2^{2^{O(k+r)}}|V(G)|^{O(1)}$.
\end{theorem}
By a well-known observation (cf.~\cite{marx-approx}), the
algorithm of Theorem~\ref{thm:mainEP} can be turned into an \FPT
approximation algorithm of the following form.
\begin{corollary}[Excluding large induced matching implies \FPT approximation]\label{cor:FPTapprox}
  Let $\cH$ be a set of graphs closed under taking induced subgraphs,
  and assume there is an integer $r\ge 1$ such that $M_r \notin
  \cH$. Then there is a polynomial-time algorithm that, given an
  instance of \maxhhdisjoint, finds a solution with $\Omega(\log \log
  OPT)$ disjoint valid paths, where $OPT$ is the maximum size of a set of
  pairwise disjoint valid paths.
\end{corollary}
Can we improve the algorithm of Theorem~\ref{thm:mainEP} to an exact
\FPT algorithm that either finds a set of $k$ disjoint valid paths or
correctly states that there is no such set? It seems that we need one
more property of $\cH$ for the existence of such algorithms. A {\em
  skew biclique of size $n+n$} is the bipartite graph $S_n$ on
vertices $a_1$, $\dots$, $a_n$, $b_1$, $\dots$, $b_n$ such that $a_i$
and $b_j$ are adjacent if and only if $i\le j$. Even though the
(hereditary closure of) the set $\cH$ of all skew bicliques has the
\EP property by Theorem~\ref{thm:mainEP} (as skew bicliques do not have
large induced matchings), disjoint paths problems with skew biclique
demand patterns can be hard. Our main result states that large induced
matchings and large skew bicliques are the only demand patterns that
make the \maxhhdisjoint problem hard.
\begin{theorem}[Main theorem: characterizing fixed-parameter tractability]\label{thm:main}
Let $\cH$ be a hereditary set of graphs. If there is an integer $r\ge 1$ such that $M_r,S_r\not\in \cH$, then \maxhhdisjoint is \FPT; otherwise, \maxhhdisjoint is \classWone-hard.
\end{theorem}
Therefore, we have obtained a tight characterization of the
fixed-parameter tractable cases of \maxhhdisjoint. Observe that the
algorithmic part of Theorem~\ref{thm:main} covers the \FPT cases we
discussed above: if the vertices in every $H\in \cH$ have at most $c$
different neighborhoods, then $\cH$ cannot contain every matching and
every skew biclique. However, Theorem~\ref{thm:main} gives some more
general \FPT cases as well: for example, if every graph in $\cH$ is a
biclique minus a matching of arbitrary size, then clearly there are no
large induced matchings or skew bicliques in $\cH$, but the number of
different neighborhoods can be arbitrarily large.  Observe also that
Corollary~\ref{cor:FPTapprox} and Theorem~\ref{thm:main} exhibit a
large class of problems that are \classWone-hard, but admit an \FPT
approximation: if $\cH$ contains every skew biclique $S_r$, but does
not contain some matching $M_r$, then \maxhhdisjoint is such a
problem. There is only a handful of known problems with this property
(see
\cite{marx-approx,DBLP:conf/icalp/GroheG07,DBLP:conf/iwpec/ChitnisHK13}),
thus this may be of independent interest.

\textbf{Our techniques.}  The first observation in the proof of
Theorem~\ref{thm:mainEP} is that if there is a small set $Z$ of
vertices such that more than one component of $G-Z$ contains valid
paths, then we can solve the problem recursively. Therefore, we may
assume that the valid paths are quite intertwined, giving us a notion
of connectivity similar to tangles. Our first goal is to find a certain
number of pairs $(s_1,t_1)$, $\dots$, $(s_h,t_h)$ such that $s_i$ and
$t_i$ are adjacent in $H$ and the set
$\{s_1,\dots,s_h,t_1,\dots,t_h\}$ is highly connected in our notion of
connectivity. In particular, the connectivity ensures that there are
many disjoint paths between $\{s_1,\dots,s_h\}$ and
$\{t_1,\dots,t_h\}$. This is not quite what we need: all we know is
that $s_i$ and $t_i$ are adjacent in $H$, but we have no information
about the adjacency of $s_i$ and $t_j$ for $i\neq j$. This is the
point where we exploit the assumption that there are no large induced
matchings in $H$. A simple Ramsey-type argument shows that if a
graph has a large (not necessarily induced) matching, then it either
has a large induced matching or a large biclique. By assumption, there
is no large induced matchings in $H$, which means that $H$ contains a
large biclique on the vertices $\{s_1,\dots,s_h,t_1,\dots,t_h\}$. Then
by the connectivity of this set, we can realize $k$ disjoint paths
with endpoints in this biclique.

The fixed-parameter tractability part of Theorem~\ref{thm:main} is
proved the following way. First, we bootstrap the algorithm with the
approximation of Theorem~\ref{thm:mainEP}: we obtain either $k$
disjoint valid paths (in which case we are done) or a set $Z$ of bounded
size covering every valid path. In the latter case, we solve the problem
by analyzing the components of $G-Z$: as there are no  valid paths
in any component $C$ of $G-Z$, essentially what we need to understand
is how subsets of terminals in $C$ can be connected to $Z$. However,
each component of $G-Z$ can contain a large number of terminals and
there can be a large number of components of $G-Z$. First, in each
component $C$ of $G-Z$, we reduce the number of terminals so that
their number is bounded: we identify terminals that are {\em
  irrelevant,} that is, we can prove that if there is a solution, then
there is a solution not using these terminals. To identify irrelevant
terminals, we use the concept of {\em representative sets,} which were
already used in the design of \FPT algorithms, mostly for path and
matroid problems
\cite{MR87a:05097,DBLP:journals/tcs/Marx09,DBLP:journals/corr/FominLPS14,DBLP:conf/soda/FominLS14,DBLP:journals/corr/ShachnaiZ14}. While
the concept is the same as in previous work, the reason why we can
give a bound on the size of representative sets is very different: as
shown by a simple Ramsey-type argument, it is precisely the lack of
large induced matchings and skew bicliques in $\cH$ that makes the
argument work. (More precisely, we need to exclude large cliques as
well, but we have a separate argument for that.)  Our algorithm can be
seen as a generalization of the ideas in the data structure of
Monien~\cite{MR87a:05097}, but it does not use any of the more
advanced matroid-based techniques of more recent
work~\cite{DBLP:journals/tcs/Marx09,DBLP:journals/corr/FominLPS14,DBLP:conf/soda/FominLS14,DBLP:journals/corr/ShachnaiZ14}. After
reducing the number of terminals to a constant in each component of
$G-Z$, next we use elementary arguments to show that every terminal in
all but a bounded number of components is irrelevant. Thus we have a
bound on the total number of terminals and then we can use the
algorithm of Robertson and Seymour \cite{MR97b:05088} on every set of
$k$ pairs of terminals.

The hardness part of Theorem~\ref{thm:main} states \classWone-hardness for
infinitely many classes $\cH$. However, we need to prove only the following two 
concrete \classWone-hardness results: when the pattern is a matching and when the pattern is a skew biclique. We prove these hardness result in a slightly stronger form: the supply graph $G$ is restricted to be planar and we show that the problems are hard even when parameterized by both the number of paths $k$ to be found and 
the treewidth $w$ of the supply graph, that is, even an algorithm with running time $f(k,w)\cdot n^{O(1)}$ seems unlikely.
\begin{theorem}[Hardness for matchings]\label{th:matchinghardintro}
If $\cH$ contains $M_r$ for every $r\ge 1$, then \maxhhdisjoint is \classWone-hard with combined parameters $k$ and $w$ (where $w$ is the treewidth of $G$), even when restricted to instances where $G$ is planar.
\end{theorem}
\begin{theorem}[Hardness for skew bicliques]\label{th:skewhardintro}
If $\cH$ contains $S_r$ for every $r\ge 1$, then \maxhhdisjoint is \classWone-hard with combined parameters $k$ and $w$ (where $w$ is the treewidth of $G$), even when restricted to instances where $G$ is planar.
\end{theorem}
Note that Theorem~\ref{th:skewhardintro} actually
implies Theorem~\ref{th:matchinghardintro}: if $\cH$ contains the
matching $M_r$ for every $r\ge 1$, then it is easy to simulate any
demand pattern, including skew bicliques. The reduction is as
follows. First, if vertex $v$ has degree $d$ in $H$, then let us
attach $d$ degree-1 neighbors to $v$ and make them terminals. Then
replace each edge $(x,y)$ of $H$ with an edge connecting a degree-1
neighbor of $x$ and a degree-1 neighbor of $y$ not incident to any
demand edge yet. This way the new demand graph becomes a matching of
$|E(H)|$ edges. Therefore, giving a separate proof for Theorem~\ref{th:matchinghardintro} is redundant. Nevertheless, we give a self-contained \classWone-hardness proof of \maxdisjoint with no restriction on the demand pattern, which, by the reduction described above, proves Theorem~\ref{th:matchinghardintro} (but not Theorem~\ref{th:skewhardintro}). We believe that the \classWone-hardness of \maxdisjoint can be already of independent interest and the proof is much simpler and cleaner than the highly technical proof of Theorem~\ref{th:skewhardintro}.

\iffull We show that \maxdisjoint is \classWone-hard with a fairly
standard parameterized reduction. We make the reduction carefully so
that the created supply graphs $G$ have treewidth bounded by a
function of $k$ and planar. This shows that the basic algorithmic
ideas of \disjointpath exploiting bounded treewidth and planarity are
unlikely to work for the more general \maxdisjoint problem.  To ensure
planarity, we reduce from the \gridtiling problem, which is a standard
technique for planar \classWone-hardness proofs (see, e.g.,
\cite{DBLP:conf/icalp/Marx12,DBLP:conf/stacs/MarxP14,DBLP:conf/soda/ChitnisHM14}).

If $\cH$ consists of every skew biclique, then we get a variant of the
problem that we call \skewdisjoint: given pairs $(s_1,t_1)$, $\dots$,
$(s_m,t_m)$, a path is valid if it connects $s_i$ and $t_j$ for some
$i\le j$. Again, by a reduction from \gridtiling, we show that
\skewdisjoint is \classWone-hard on planar graphs of treewidth
bounded by a function of $k$.  However, this time the gadget construction is more involved,
as the dense demand pattern makes the problem less amenable to the
implementation of independent choices needed in gadgets.

Let us point out that, by Corollary~\ref{thm:main}, \skewdisjoint is
one of those apparently rare concrete problems that are
\classWone-hard, but admit an \FPT approximation.
\fi

\iffull
\textbf{An alternate formulation of the results.}  Let us discuss a
different formulation of our results, which is somewhat more limited,
but perhaps reveals more precisely the nature of the problem. Recall
that the motivation for studying hereditary classes comes from the
fact that removing a vertex $v$ from $H$ can be easily expressed by
assigning $v$ to an isolated vertex of the graph. We can consider another
operation that is easy to simulate: identifying an independent set $S$
of $H$ into a single vertex (that is, we obtain $H'$ from $H$ by
removing $S$ and introducing a new vertex $v$ that is adjacent to
every neighbor of $S$ in $V(H)\setminus S$). Given an instance of
\maxdisjoint with demand pattern $H'$, we can simulate it with demand
pattern $H$ by attaching $|S|$ new degree-1 vertices to $v$ and
assigning $S$ to these vertices in an arbitrary way. It is easy to see
that the two instances are equivalent. Therefore, intuitively, we can
say that adding every $H'$ to $\cH$ that arises from identifying an
independent set in some $\cH$ should not make the problem harder. We
still have the same technical caveats as before, such as the
difficulty of finding a suitable $H$ given $H'$, but it seems closer
to the spirit of the problem if we consider hereditary classes $\cH$
closed also under identifying independent sets. Observe that if such a class
contains arbitrarily large matchings, then every graph appears in the
class: every graph with $m$ edges can be obtained from the matching
$M_m$ by identifying independent sets in an appropriate way. Therefore, our classification can be stated in a very compact way for such classes.
\begin{theorem}[Main result, alternate formulation]\label{thm:mainalt}
Let $\cH$ be a hereditary class of graphs closed under identifying independent sets.
\begin{itemize}
\item If $\cH$ does not contain every skew biclique, then \maxhhdisjoint is \FPT.
\item If $\cH$ contains every skew biclique, but does not contain every graph, then \maxhhdisjoint is \classWone-hard, admits an \FPT-approximation, and has the \EP property.
\item If $\cH$ contains every graph, then \maxhhdisjoint is \classWone-hard
  and does not have the \EP property.
\end{itemize}
\end{theorem}
Therefore, it is actually only the skew bicliques that prevent the
problem from being \FPT, and there we have the \FPT-approximation and
the \EP property for every nontrivial restriction of the problem.

\fi

\iffull
\textbf{Notation.} We conclude the section with some notation.  We will use the notation $H \subseteq G$ to indicate that a graph $H$ is a subgraph of a graph $G$.  Given two subgraphs $H_1$ and $H_2$ of a graph $G$, the graph $H_1 \cup H_2$ has vertex set $V(H_1) \cup V(H_2)$ and edge set $E(H_1) \cup E(H_2)$.  Similarly, the graph $H_1 \cap H_2$ has vertex set $V(H_1) \cap V(H_2)$ and edge set $E(H_1)\cap E(H_2)$.  We will use $|G|$ as shorthand notation for $|V(G)|$.  A \emph{separation} in a graph $G$ is a pair $(X, Y)$ of edge-disjoint subgraphs such that $X \cup Y = G$.  The separation is trivial if $V(X) \subseteq V(Y)$ or $V(Y) \subseteq V(X)$.  The \emph{order} of the separation is $|X \cap Y|$.  We will use $G[U]$ to indicate the subgraph induced on the subset $U$ of vertices.  Occasionally, the set $U$ will contain elements not in the set $V(G)$; in this case, $G[U]$ refers to the graph $G[V(G) \cap U]$.  We will denote by $G-U$ the subgraph $G[V(G) \setminus U]$.  For a subgraph $H$ of $G$, $G-H$ denotes the subgraph $G - V(H)$.   

\fi


\iffull
\section{Excluding induced matchings: Erd\H os-P\' osa property and FPT approximation}
\label{sec:approx}

In this section, we give the proof of Theorem \ref{thm:mainEP}.  We begin with a more technical statement which will facilitate the recursive step of the algorithm.

\begin{restatable}[Excluding large induced \allowbreak matching implies \EP property]{\retheorem}{restatemainEP}
\label{thm:mainEP_revised}
Let $G$ be a graph, $T \subseteq V(G)$, $k,r\ge 1$ integers, and $H$ a graph with $V(H) = T$.  Assume that $T$ is an independent set and $\deg_G(v)= 1$ for all $v \in T$.  There exists an algorithm which takes as input $G$, $T$, $k$, $r$, and $H$ and returns one of the following:
\begin{enumerate}
\item $k$ pairwise disjoint valid paths, or
\item a set $X$ of at most $4\cdot 5^{20(k+r)}$ vertices such that every valid path intersects $X$.
\item a subset $Z \subseteq T$ with $|Z| = 2r$ such that $H[Z]$ is an induced matching.
\end{enumerate}
Moreover, the algorithm runs in time $4^{3\cdot 5^{10(k+r)}}|V(G)|^{O(1)}$.
\end{restatable}

Theorem \ref{thm:mainEP} follows easily from Theorem \ref{thm:mainEP_revised}.  

\begin{proof}[Proof (of Theorem \ref{thm:mainEP} assuming Theorem \ref{thm:mainEP_revised}).]
Let $G$, $H \in \zF$, $T$, $k$ be given.  Assume $M_r$ is not contained in $\zF$ for some positive integer $r$.  We construct an auxiliary graph $G'$ by adding a new vertex $x'$ to the graph adjacent only to $x$ for every vertex $x \in T$.  Let $T' = \{x': x \in T\}$, and let $H'$ be the copy of $H$ on $T'$.  Then $G'$ has $k$ pairwise disjoint valid paths if and only if $G$ has $k$ pairwise disjoint valid paths.  Similarly, if $Z$ is a set in $G'$ intersecting all the valid paths, then $(Z \setminus T') \cup \{x\in T: x' \in Z \cap T\}$ is a set in $V(G)$ intersecting all the valid paths in $G$.  The theorem now follows by Theorem \ref{thm:mainEP_revised} and the assumption that $H$ has no induced subgraph isomorphic to $M_r$.  
\end{proof}

The proof of Theorem \ref{thm:mainEP_revised} will occupy the remainder of the section; we outline how the proof will proceed.  Consider for a moment a more general problem.  Assume we are trying to show that the \EP property holds for a set $\zC$ of connected graphs: i.e. that there exists a function $f$ such that for every positive integer $k$ and graph $G$, either $G$ has $k$ disjoint subgraphs in $\zC$ or there exists $f(k)$ vertices intersecting every subgraph of $G$ in $\zC$.  If we consider a minimal counterexample, then there cannot exist a separation $(X, Y)$ of small order such that each of $X$ and $Y$ contain a subgraph in $\zC$.  Otherwise, by minimality, we can either find $k-1$ disjoint $\zC$-subgraphs in $X-Y$ or a set of $f(k-1)$ vertices in $X-Y$ intersecting all such subgraphs.  If we found $k-1$ subgraphs, along with the graph in $Y$, we would have $k$ subgraphs in $\zC$, contradicting our choice of counterexample.  Thus, we may assume there is hitting set $Z_X$ of size $f(k-1)$ intersecting every $\zC$-subgraph in $X-Y$.  Similarly, there exists a bounded hitting set $Z_Y$ in $Y-X$.  By our assumption that $\zC$ consists of only connected subgraphs, every subgraph of $G$ in $\zC$ must be contained in either $X$ or $Y$.  Thus, $Z_X \cup Z_Y \cup V( X \cap Y)$ is a hitting set of all $\zC$-subgraphs in $G$ of size $2f(k-1) + |X \cap Y|$.  If the function $f$ grows sufficiently quickly, this will yield a contradiction.  

The conclusion is that for every small order separation $(X, Y)$, only one of $X$ or $Y$ can contain a subgraph in $\zC$.  This defines a \emph{tangle} in the graph $G$.  Tangles are a central concept in the Robertson-Seymour theory of graph minors \cite{Robertson1991153}. We will not need the exact definitions here, as we do not use any technical tangle results.  However, this argument shows how tangles arise naturally in proving \EP type results; see \cite{W11} for another example.  The proof of Theorem \ref{thm:mainEP_revised} is not presented in terms of tangles for two reasons.  First, the tangle defined above only exists in a minimal counterexample to the theorem.  While this suffices for an existential proof of an \EP bound, we are also interested in an algorithm.  We need to consider all possible problem instances and then we will not always have such a tangle to work with.  Second, the proof does not use any technical tangle theorems; in the interest of simplicity of the presentation, we do not introduce tangles although they inform and motivate how the proof proceeds.

Now return to the specific problem at hand.  Consider a graph $G$, $k$, $r$, $T \subseteq V(G)$, and demand graph $H$ with $V(H) = T$.  A subset $X \subseteq T$ is \emph{well-linked} if for any $U, W \subseteq X$ with $|U| = |W|$, there exist $|U|$ disjoint paths from $U$ to $W$. We attempt to find a large subset $T' \subseteq T$ such that 
\begin{enumerate}
\item $H[T']$ contains a perfect matching, and
\item $T'$ is well-linked in $G$.
\end{enumerate}

The significance of the perfect matching in $H[T']$ in 1.~above is that it will allow us to apply Ramsey's theorem to find a useful subgraph of $H$.  
\begin{lemma}[Ramsey's Theorem]\label{lem:multiramsey}
  Let $c$, $r$, and $n$ be positive integers with $n\ge
  c^{rc}$. Given a $c$-coloring of the edges of an $n$-clique, we
  can find in polynomial time a monochromatic $r$-clique in the
  coloring.
\end{lemma}
From Ramsey's theorem, we show that the graph $H[T']$ which contains a perfect matching must contain either an induced subgraph which is a matching or a complete bipartite subgraph.
\begin{lemma}\label{lem:ramsey}
Let $H$ be a graph and $r \ge 1$ a positive integer.  If $H$ contains $M_{5^{10r}}$ as a subgraph, then $H$ either contains $M_r$ as an induced subgraph or $H$ contains $K_{r,r}$ as a subgraph.  Moreover, we can find the desired subgraph in polynomial time.
\end{lemma}

\begin{proof}
By Lemma \ref{lem:multiramsey}, every clique with $c:=5^{10r}$ vertices such that the edges are colored one of five colors contains a clique subgraph of size $2r$ where all the edges are the same color.    

Let $H$ contain $M_c$ as a subgraph, and let $\{x_i, y_i\}$ for $1 \le i \le c$ form the edges of the matching.  Consider the clique on $c$ vertices, with the vertices labeled $1, \dots, c$.  We define a 5-coloring of the edges as follows.  For an edge of the clique $ij$ with $i < j$, we color the edge:
\begin{enumerate}
\item color 1 if no edge of $H$ has one end in $\{x_i, y_i\}$ and one end in $\{x_j, y_j\}$, 
\item color 2 if $x_i$ is adjacent $x_j$, 
\item color 3 if $y_i$ is adjacent $y_j$ and $x_i \nsim x_j$, 
\item color 4 if $x_i$ is adjacent $y_j$ and $x_i \nsim x_j$, $y_i \nsim y_j$, 
\item color 5 if $y_i$ is adjacent $x_j$ and $x_i \nsim x_j$, $y_i \nsim y_j$, $x_i \nsim y_j$
\end{enumerate}
where all adjacencies are in the graph $H$.  This defines a 5-coloring of the edges of the clique.  By our choice of $c$, there exists a subset of vertices of size $2r$ inducing a monochromatic subclique, and we can identify it in polynomial time.  Without loss of generality, we may assume that the vertices $1 \le i \le 2r$ of the clique induce such a monochromatic clique.  If the subclique has color 1, then $H$ contains an induced matching of size $2r$.  If the monochromatic clique has color 2 or 3, then $H$ contains a clique subgraph of size $2r$.  Finally, if the subclique has color 4 (respectively, 5), then the vertices $\{x_1, \dots, x_r\} \cup \{y_{r+1}, \dots, y_{2r}\}$ (respectively, $\{y_1, \dots, y_r\} \cup \{x_{r+1}, \dots, x_{2r}\}$) induce a $K_{r, r}$ subgraph of $H$.
\end{proof}

Given a large set $T'$ satisfying 1 and 2 above, the argument is fairly straightforward.
By Lemma \ref{lem:ramsey}, either $H[T']$ contains an induced matching of size $r$ or there exist two sets $U$ and $W$ in $T'$, each of size $k$, such that every vertex in $U$ is adjacent every vertex in $W$ (in $H$).  As we are assuming $T'$ is well-linked in $G$, given such a $U$ and $W$, we can find $k$ disjoint paths from $U$ to $W$ and these will necessarily be valid paths.

How can we find such a subset $T'$ of $T$?  It is easy to find such a subset $T'$ of size two --- take two vertices in $T$ which are adjacent in $H$ and connected by a path.  Thus, the difficulty will lay in showing we can find $T'$ sufficiently large to apply the desired Ramsey argument.  Note that the property of being well-linked is a standard certificate that a graph has a large tangle.  We proceed by effectively showing that we either have a tangle as in a minimal counterexample to the \EP property, or alternatively, finding a separation separating two valid paths and then recurse on the smaller graphs.  More explicitly, we replace property 2 above by:
\begin{itemize}
\item[$2'.$] there does not exist a separation $(X, Y)$ of order $<|T'|$ with $T' \subseteq V(X)$ and $Y-X$ containing a valid path.
\end{itemize}
Property $2'$ forces a similar behavior to well-linkedness in a tangle without requiring the technical properties of a tangle.  We show that either we can grow $T'$ by two vertices and satisfy $1$ and $2'$, or alternatively find a separation where we can recurse.  The problem of identifying separations which separate $T'$ from a valid path leads us to introduce the notions of $\zP$-tight separations and $\zP$-free sets.  We define these notions rigororously in Subsection \ref{sec:tight2} and present efficient algorithms for finding them.  We conclude with the proof of Theorem \ref{thm:mainEP_revised} in Subsection \ref{sec:EPfinal}.

\iffull

\subsection{$\zP$-tight separations and $\zP$-free sets}\label{sec:tight2}

In this subsection, we give an algorithm for finding what we call tight separations.  We begin with the definition.

\begin{definition}[$\zP$-tight]
Let $G$ be a graph and $T \subseteq V(G)$.  Let $\zP$ be a set of connected subgraphs in $G-T$.  
We say a separation $(U, W)$ is \emph{$\zP$-tight for $T$} if
\begin{itemize}
\item[i.] $T \subseteq V(U)$;
\item[ii.] there exists $P \in \zP$ with $P \subseteq W-U$;
\item[iii.] there does not exist a separation $(U', W')$ and element $P \in \zP$ with $|U' \cap W'| \le |U \cap W|$, $U \subsetneq U'$, and  $P \subseteq W'-U'$.
\end{itemize}
When there can be no confusion as to the set $\zP$, we will simply say a separation is \emph{tight} for $T$.
\end{definition}
Thus a separation is tight if the portion not containing $T$ is made as small as possible while not increasing the order of the separation and maintaining the property that it still contains an element of $\zP$.  Note that a tight separation may have order greater than $|T|$.  

Given a graph $G$, $T \subseteq V(G)$, and $\zP$ a non-empty set of connected subgraphs of $G-T$, there always exists a tight separation of order at most $|T|$.   To see this, let $(U, W)$ be a separation of minimum order satisfying $i$ and $ii$, and subject to that, to maximize $|V(U)| + |E(U)|$.  Such a separation always exists as the trivial separation $(T, G)$ satisfies $i$ and $ii$ with $T$ treated as the graph with vertex set $T$ and no edges. Then $(U, W)$ will be of order at most $|T|$ and satisfy $i - iii$.  The same argument shows the following observation.

\begin{OB}\label{ob1}
Let $G$ be a graph, $T \subseteq V(G)$, 
and $\zP$ a non-empty set of connected graphs in $G-T$.  
Let $(U, W)$ be a separation satisfying $i$~and $ii$ in the definition of tight for $T$.  
There exists a separation $(U', W')$ of order at most $|U \cap W|$ which is tight for $T$ and $U \subseteq U'$.  
\end{OB}

We now turn our attention to finding a tight separation when given a graph $G$, subset $T$ of vertices, and set of connected subgraphs $\zP$.  If we were given $\zP$ as a list of subgraphs, one could use standard flow algorithms to find a minimum order separation separating the terminals $T$ from each element of $\zP$.   However, in the applications to come, we will not have any reasonable bound on the size of $\zP$ (in terms of $|V(G)|$).  Thus, we assume $\zP$ is given by an oracle and bound the runtime in the size of the terminal set $T$.  The difficulty now lays in identifying an appropriate set of separations to check as potential candidates for a tight separation.  To find such a set of separations, we use what are called important separators in a graph. 

\begin{definition}[separator]
Let $G$ be an undirected graph and let $X,Y \subseteq V(G)$ be two disjoint sets. A set $S\subseteq V(G)$ of vertices is an \emph{$X-Y$ separator} if $S$ is disjoint from $X \cup Y$ and there is no component $K$ of $G - S$ with both $V(K) \cap X \neq \emptyset$ and $V(K)  \cap Y \neq \emptyset$.
\end{definition}

\begin{definition}[important separators]
Let $X, Y \subseteq V (G)$ be disjoint sets of vertices, $S \subseteq V (G)$ be an $X - Y$ separator, and let $K$ be the union of the vertex sets of every component of $G - S$ intersecting $X$. We say that $S$ is an important $X - Y$ separator if it is inclusionwise minimal and there is no $X - Y$ separator $S'$ with $|S'| \le |S|$ such that $K' \supsetneq K$, where $K'$ is the union of every component of $G - S'$ intersecting $X$.
\end{definition}

\begin{lemma}[Finding important separators~\cite{MR}]\label{lem:impsep}
Let $X, Y \subseteq V (G)$ be disjoint sets of vertices in a graph $G$. For every $p \ge 0$, there are at most $4^p$ important $X - Y$ separators of size at most $p$. Furthermore, we can enumerate all these separators in time $4^p \cdot p \cdot (|E(G)| + |V (G)|)$.
\end{lemma}

As a first step to presenting an algorithm for finding a $\zP$-tight separation, we give an algorithm for testing whether a set is \emph{$\zP$-free}.    
\begin{definition}[$\zP$-free]
Let $G$ be a graph, $T \subseteq V(G)$, and $\zP$ a set of connected subgraphs of $G-T$.  The set $T$ is \emph{$\zP$-free} if there does not exist a separation $(U, W)$ of order strictly less than $|T|$ and $P \in \zP$ such that $T \subseteq U$ and $V(P) \subseteq V(W-U)$.
\end{definition}
Let $G$ be a graph, $T \subseteq V(G)$, and $\zP$ a set of connected subgraphs of $G-T$.  We will show that there is an algorithm for efficiently testing whether $T$ is $\zP$-free or not for sets $T$ of bounded size.  We will typically assume that $\zP$ is given by an oracle.  A \emph{$\zP$-oracle} is a function $f$ such that for any subgraph $H \subseteq G$, $f$ responds ``yes" if there is an element $P \in \zP$ such that $P \subseteq H$ and ``no" otherwise.  A \emph{certificate} that $T$ is not free is a separation $(X, Y)$ of order strictly less than $|T|$ such that $T \subseteq V(X)$ and there exists $P \in \zP$ with $P \subseteq Y-X$.

\defparproblemu{{\sc Test $\zP$-Free}}
{
A graph $G$, $T \subseteq V(G)$, $\zP$-oracle $f$ for a set $\zP$ of connected subgraphs of $G-T$.
}{}
{
either 
\begin{itemize}
\item confirm that $T$ is $\zP$-free or 
\item output a separation $(X, Y)$ which is a certificate that $T$ is not free; moreover, $(X, Y)$ is of minimum order among all such separations.
\end{itemize}

}

\begin{lemma}[Testing if a set is $\zP$-free]\label{lem:freetest}
There exists an algorithm solving {\sc Test $\zP$-Free} running in time $4^{|T|}|V(G)|^{O(1)}$ utilizing $O(|V(G)|4^{|T|})$ calls of the $\zP$-oracle.
\end{lemma}
\begin{proof}
Let $G$, $T$, and an oracle $f$ for the set $\zP$ be given.  Let $n = |V(G)|$ and $m = |E(G)|$.  There is a slight technical issue which we must address.  We want to proceed by calculating all separations $(X, Y)$ where $T \subseteq X$ and $X \cap Y$ is an important separator for some vertex $y \in Y$.  However, we will additionally need to consider such separations where $X \cap Y$ intersects the set $T$.  However, in the definition of important separator, we do not consider separators which intersect one of the two sets.  Thus, we define an auxiliary graph $G'$ formed by adding a new vertex $a'$ adjacent to every vertex of $T$ and consider important separators separating a vertex from $a'$ in $G'$. 

Fix a vertex $y \in V(G) \setminus T$.  Enumerate all important $y - a'$ separators in $G'$ of size at most $|T| - 1$.  For each separator $S$, let $K_S$ be the component of $G'-S$ containing $y$.  Using the $\zP$-oracle $f$, check if there exists an element $P \in \zP$ with $P \subseteq K_S$.  We do this for every $y \in V(G) \setminus T$.  By Lemma \ref{lem:impsep}, this can be done in time $O(4^{|T|} |T| (n+m) n\cdot m)$ with at most $n4^{|T|}$ calls to to the $\zP$-oracle, as desired.  

Assume, as a case, we find a vertex $y \in V(G) \setminus T$ and an important $y - a'$ separator $S$ such that the subgraph induced by $K_S$ contains an element of $\zP$.  Pick $y$ and $S$ over all such vertices and important separators to minimize $|S|$.  We return the separation $(G - V(K_S), G[V(K_S) \cup S] - E(G[S]))$ as a certificate that $T$ is not $\zP$-free. If we find no such important separator, we return that $T$ is $\zP$-free.  

To see correctness, first observe that if we return a separation, it must be the case that $T$ is not $\zP$-free.  Thus, we must only show that if $T$ is not $\zP$-free, we correctly find a minimum order separation certifying so.  Assume that $T$ is not free, and let $(X, Y)$ be a separation such that:
\begin{itemize}
\item[$i.$] $T \subseteq V(X)$ and there exists $P \in \zP$ such that $P \subseteq Y-X$.
\item[$ii.$] Subject to $i$, the size of $|X \cap Y|$ is minimized.
\item[$iii.$]  Subject to $i$ and $ii$, $|Y|$ is minimized.
\end{itemize}
Moreover, assume that the algorithm finds no important separator $S$ of order at most $|X \cap Y|$ such that $S$ separates $a'$ from an element of $\zP$.  
Note, by $iii$, we may assume that $Y - X$ is connected.  

Fix a vertex $y \in Y-X$.  We considered all important $z - a'$ separators in $G'$ for every vertex $z \in V(G) \setminus T$ and did not find an important separator of order at most $|X \cap Y|$ which separated an element of $\zP$ from $a'$.  Specifically, it cannot be the case that $X \cap Y$ is an important $y - a'$ separator. 
By our choice of $(X, Y)$ to satisfy $ii$ and the observation that $Y - X$ is connected, we have that $X \cap Y$ is a minimal (by containment) $y-a'$ separator.  We conclude that there exists an important $y - a'$ separator $S$ of order $|X \cap Y|$ such that the component of $G'-S$ containing $y$ contains all of $Y-X$.   Note that here we are again using the fact that $Y-X$ is connected.  Thus, the separator $S$ separates $a'$ from an element of $\zP$, contradicting our assumptions.  This completes the proof.
\end{proof}

We now turn our attention to the algorithm for finding a $\zP$-tight separation.

\defparproblemu{{\sc Find $\zP$-tight}}
{
A graph $G$, $T \subseteq V(G)$, $\zP$-oracle $f$ for a set $\zP$ of connected subgraphs of $G-T$.
}{}
{
A separation $(X, Y)$ of order at most $|T|$ which is $\zP$-tight for the pair $(G, T)$ and of minimum order among all such tight separations.
}

\begin{lemma}[Finding a $\zP$-tight separation]\label{lem:tightfind}
There exists an algorithm solving {\sc Find $\zP$-tight} running in time $4^{|T|}n^{O(1)}$ utilizing $O(|T|\cdot|V(G)|^24^{|T|})$ calls of the $\zP$-oracle.
\end{lemma}

\begin{proof}
Let $G$, $T \subseteq V(G)$, and a $\zP$-oracle $f$ for a set of connected subgraphs $\zP$ in $G$ be given.  Observe that for any $X \subseteq V(G)$, the function $f$ is a $\zP'$-oracle for the subset $\zP' \subseteq \zP$ of elements of $\zP$ contained in the subgraph $G[X]$.  

We first use the algorithm given in Lemma \ref{lem:freetest} to check if $T$ is $\zP$-free.  If $T$ is not free, let $(X_1, Y_1)$ be the separation returned by the algorithm.  If $T$ is free, let $(X_1, Y_1)$ be the trivial separation $(T, G)$ with $T$ treated as the graph with vertex set $T$ and no edges.  Let $\zP_1 = \{P \in \zP: P \subseteq Y_1\}$.  Note that $X_1 \cap Y_1$ is $\zP_1$-free in $Y_1$ by the guarantee that $(X_1, Y_1)$ is a minimum order separation separating $T$ from an element of $\zP$.

We now define inductively define separations $(X_i, Y_i)$ with the following properties.
\begin{enumerate}
\item $(X_i, Y_i)$ is a separation of $G$ of order $|X_1 \cap Y_1|$ with $V(X_{i-1}) \subsetneq V(X_i)$.
\item There exists $P \in \zP$ with $P \subseteq Y_i$.
\end{enumerate}

Given $(X_i, Y_i)$, for $i = 1, \dots, k$, we now describe how to either construct $(X_{k+1}, Y_{k+1})$ or determine that $(X_k, Y_k)$ satisfies the desired properties for the output.  

First, consider the case when $Y_k - (X_k \cap Y_k)$ has multiple connected components. Let $C$ be a component of $Y_k - (X_k \cap Y_k)$ such that $C$ contains an element of $\zP$.  Then the separation $(X_{k+1}, Y_{k+1}) = (G - C, G[V(C) \cup V(X_i\cap Y_i)] - E(G[X_i \cap Y_i]))$ satisfies 1 and 2.

Assume now that $Y_k - (X_k \cap Y_k)$ has exactly one component.  For every $x \in V(X_k \cap Y_k)$, $x$ is adjacent to a vertex of $Y_k - X_k$ by the fact that no smaller order separation separates $T$ from an element of $\zP$.  Arbitrarily fix a neighbor $x'$ of $x$ in $Y_k - X_k$.  We apply the algorithm of Lemma \ref{lem:freetest} on the graph $Y_k$, subset of vertices $V(X_k \cap Y_k) \cup \{x'\}$, and the set $\zP'= \{P \in \zP: P \subseteq Y_k - (V(X_k \cap Y_k) \cup \{x'\})\}$ of connected subgraphs.  Assume we get a separation $(X', Y')$ certifying that $V(X_k \cap Y_k) \cup \{x'\}$ is not $\zP'$-free in $Y_k$.  It must hold that $(X', Y')$ is of order exactly $|X_k \cap Y_k|$ and that there is an element $P \in \zP$ such that $P \subseteq Y'-X'$.  Thus, $(X_{k+1}, Y_{k+1}) = (X_k \cup X', Y')$ satisfies 1 and 2 above.  We arbitrarily fix $(X_{k+1}, Y_{k+1})$ among all such possibilities and continue.  

To define $(X_{k+1}, Y_{k+1})$ given $(X_k, Y_k)$ takes at $O(4^{|T|} |T|^2 (n+m) n \cdot m)$ time and at most $n\cdot |T| \cdot 4^{|T|}$ calls to to the $\zP$-oracle.  As $|V(X_{k+1})| > |V(X_k)|$, in time $O(4^{|T|} |T|^2 (n+m) n^2 \cdot m)$ with at most $n^2 \cdot |T| \cdot 4^{|T|}$ calls to to the $\zP$-oracle, we find $(X_k, Y_k)$ such that $Y_k - (X_k \cap Y_k)$ has exactly one component and for all $x \in X_k \cap Y_k$, the set $V(X_k \cap Y_k) \cup \{x'\}$ is free in $Y_k$.  Without loss of generality, we may assume that the edges of $G[V(X_k \cap Y_k)]$ are contained in $X_k$

To complete the proof, it suffices to show that $(X_k, Y_k)$ is a tight separation.  If not, there exists a separation $(U, W)$ with $X_k \subsetneq U$ and an element $P \in \zP$ with $V(P) \subseteq V(W-U)$.  Note that the order of $(U, W)$ must be the same as $(X_k, Y_k)$ and $V(X_k) \subsetneq V(U)$.  We have that $U \cap W \neq X_k \cap Y_k$, lest $Y_k - (X_k \cap Y_k)$ have multiple components.  Thus there exists a vertex $x \in X_k \cap Y_k$ which is contained in $U - W$.  The separation $(U, W)$ contradicts the fact that $V(X_k \cap Y_k) \cup \{x'\}$ is free in $Y_k$.  This completes the proof of the lemma.
\end{proof}

\subsection{Proof of Theorem \ref{thm:mainEP_revised}}\label{sec:EPfinal}

Before proceeding with the proof, we will need several technical results.  

\begin{lemma}\label{lem:tight1}
Let $G$ be a graph, $T, T' \subseteq V(G)$ with $T' \subseteq T$.  Let $\zP$ a set of connected subgraphs in $G-T$.  Assume that $T'$ is $\zP$-free and let $t = |T'|$.
Let $(U', W')$ be a $\zP$-tight separation for $T'$ of order $t$, and let $(U_1, W_1)$ and $(U_2, W_2)$ be distinct $\zP$-tight separations for $T$, each of order $t+1$.  
Then one of the following holds:
\begin{enumerate}
\item $V(U'\cap W') \cup V(U_1 \cap W_1) \cup V(U_2 \cap W_2)$ is a hitting set for $\zP$.
\item There exists $P \in \zP$ such that $P$ is contained in one of the graphs $U'$, $U_1$, or $U_2$.
\item $V(U_1) \cap V(U_2) = V(U')$.  
\end{enumerate}
\end{lemma}

\begin{proof}
We may assume there exists $P \in \zP$ which is disjoint from the set $V(U' \cap W') \cup V(U_1 \cap W_1) \cup V(U_2 \cap W_2)$.  Lest we satisfy 2, we may assume as well that $P \subseteq W_1\cap W_2 \cap W'$.  Note by construction that $P$ is disjoint from $U_1 \cup U_2 \cup U'$.

We first show that $V(U') \subseteq V(U_1 \cap U_2)$.  Fix $i \in \{1, 2\}$.    We show that $U' \subseteq U_i$.  Consider the two separations $(U' \cap U_i, W' \cup W_i)$ and $(U' \cup U_i, W' \cap W_i)$.
The sum of the orders of the separations $(U' \cap U_i, W' \cup W_i)$ and $(U' \cup U_i, W' \cap W_i)$ is equal to the sum of the orders of the two separations $(U', W')$ and $(U_i, W_i)$, namely $2t+1$.  

The separation $(U' \cap U_i, W' \cup W_i)$ has $T' \subseteq V(U' \cap U_i)$.  Moreover, as the path $P$ is disjoint from  $U_1 \cup U_2 \cup U'$, it is contained in $(W'\cup W_i) - (U' \cap U_i)$.  We conclude from the fact that $T'$ is $\zP$-free that the separation has order at least $t$.  It follows that $(U' \cup U_i, W' \cap W_i)$ is a separation of order at most $t+1$ with $P \subseteq (W' \cap W_i) - (U' \cup U_i)$.  
It follows that $U' \cup U_i = U_i$ by property $iii$~in the definition of tight for $(U_i, W_i)$ and thus $U' \subseteq U_i$ as desired.

We conclude that $V(U') \cup {T} \subseteq V(U_1) \cap V(U_2)$.  The separation $(U_1 \cup U_2, W_1 \cap W_2)$ must be of order at least $t+2$, lest we violate $iii$~for one of the separations $(U_1, W_1)$ or $(U_2, W_2)$.  Note that here we are using the fact that the separations $(U_1, W_1)$ and $(U_2, W_2)$ are distinct.  It follows that $(U_1 \cap U_2, W_1 \cup W_2)$ is a separation of order at most $t$.  Consequently, $(U_1 \cap U_2 , W_1 \cup W_2)$ is a separation of order at most $t$ with $V(U') \subseteq V(U_1 \cap U_2)$.  By $iii$ in the definition of tight for $(U', W')$, we have that $V(U') = V(U_1 \cap U_2)$, completing the proof.
\end{proof}

Let $G$ be a graph.  Let $T \subseteq V(G)$ be an independent set where $\deg_G(v) = 1$ for all $v \in T$, and let $H$ be a graph with $V(H) =T$. We define the set of \emph{truncated valid paths} to be the set $$\zP = \{P -T: \text{$P$ is a valid path.}\}$$  Note that by our assumptions on $T$, every element of $\zP$ is a path.  Note as well that given $G$, $T$, and $H$, for any set $X \subseteq V(G)$, we can test whether $G[X]$ contains an element of $\zP$ in time $O(|V(G)| + |E(G)| + |E(H)|)$.  

\begin{lemma}[Growing a $\zP$-free \allowbreak set with a perfect matching]\label{lem:alg1}
Let $G$ be a graph.  Let $T \subseteq V(G)$ be an independent set where $\deg_G(v) = 1$ for all $v \in T$, and let $H$ be a graph with $V(H) =T$.  Let $\zP$ be the set of truncated valid paths.  Let $T'\subseteq T$ be a subset such that 
\begin{itemize}
\item[$i.$] $H[T']$ contains a perfect matching and 
\item[$ii.$] $T'$ is $\zP$-free.
\end{itemize}
There exists an algorithm which takes as input $G$, $T$, $H$, and $T'$ and produces in output one of the following:
\begin{enumerate}
\item a subset $Z$ of at most $|T'|(|T'|+3)$ vertices intersecting every valid path in $G$;
\item a separation $(X, Y)$ of $G$ of order at most $|T'| + 2$ such that both $X$ and $Y$ contain an element of $\zP$;
\item a subset $\bar{T}$ such that $T' \subseteq \bar{T} \subseteq T$, $|\bar{T}| = |T'| + 2$, $H[\bar{T}]$ contains a perfect matching, and $\bar{T}$ is $\zP$-free.
\end{enumerate}
The algorithm runs in time $4^{|T'|}|V(G)|^{O(1)}$.
\end{lemma}

\begin{proof}
Let $G$, $T$, $H$, and $T'$ be given.  Let $|V(G)| = n$, $|E(G)| = m$, $|E(H)| = m'$, and $|T'| = t$.  By our assumptions on $T$, $P-T$ is a (non-empty) path for every valid path $P$.  Thus, any set of vertices intersecting every element of $\zP$ also intersects every valid path.

We first find a separation $(X, Y)$ which is tight for $T'$ of minimal order.  By assumption, $(X, Y)$ has order $t$.  Lemma \ref{lem:tightfind} allows us to do this in time $4^t n^{O(1)}$.  Note, we are using here that we can test for elements of $\zP$ in time $O(n +m+m')$.  Without loss of generality, we may assume that $E(G[V(X \cap Y)])$ is contained in $E(X)$.

We can determine in time $O(n+m + m')$ if the separation $(X, Y)$ of $G$ contains an element of $\zP$ in $X$ as well as $Y$.  If so, we return the separation $(X, Y)$ satisfying 2.   
Thus, we may assume that all elements of $\zP$ intersect a vertex of $Y-X$ and at least one element of $\zP$ is contained in $Y - X$.

Let the vertices of $X \cap Y$ be $\{x_1, x_2, \dots, x_t\}$.  Each vertex $x_i$ has a neighbor in $Y - X$, lest $(X, Y - x_i)$ form a separation of order $t-1$ violating our assumption that $T'$ is free.  Arbitrarily fix $x_i'$ to be a neighbor of $x_i$ in $Y-X$ for all $i = 1, \dots, t$.  Let $\zP_i'$ be the set of elements of $\zP$ contained in $Y-(V(X) \cup x_i')$.  For each $i = 1, \dots, t$, we find a $\zP_i'$-tight separation $(U_i, W_i)$ of minimal order in $Y$ for $V(X\cap Y) \cup \{x_i'\}$ using the algorithm of Lemma \ref{lem:tightfind}.  We can do this in time $4^t n^{O(1)}$.  Note that $(U_i, W_i)$ has order $t+1$ for all $i$.

Let $Z : = \bigcup_1^t (U_i \cap W_i) \cup (X \cap Y) \cup T'$.  Then $|Z| \le t(t+1) + 2t = t(t+3)$.  If $Z$ intersects every valid path, we return $Z$ to satisfy 1.  We can check this in time $O(n+m+m')$, and therefore proceed assuming that there exists a valid path $\bar{P}$ which is disjoint from $Z$.  Fix such a path $\bar{P}$ for the remainder of the proof; let $P = \bar{P} - T$ and let $\bar{T} = T' \cup \{V(\bar{P}) \cap T\}$.  Given that the endpoints of $\bar{P}$ are $H$-adjacent, it follows that $H[\bar{T}]$ contains a perfect matching.  We test in time $4^t n^{O(1)}$ if $\bar{T}$ is $\zP$-free in $G$.  If it is, we return $\bar{T}$ satisfying 3.  

Note that $\bar{P}$ is disjoint from $X \cap Y$.  As no element of $\zP$ is contained in $X$, it follows that $\bar{P}$ is contained $Y - X$.  Specifically, the endpoints of $\bar{P}$, the vertices $\bar{T} \setminus T'$, are contained in $V(Y) \setminus V(X)$.  

Assume, to reach a contradiction, that $\bar{T}$ is not $\zP$-free.   In time $4^t n^{O(1)}$, we find a tight separation $(C, D)$ which is tight for $\bar{T}$ of minimum order.  As $T' \subseteq \bar{T}$, the separation $(C, D)$ is tight for $T'$ as well.  It follows that $(C, D)$ has order either $t$ or $t+1$.  We check in time $O(m+n +m')$ if $C$ contains an element of $\zP$.  If it does, we return $(C, D)$ as a separation satisfying 2.  Thus, we may assume that no element of $\zP$ is contained in $C$.

We check in time $O(n+m + m')$ whether $(C \cap D) \cup (X \cap Y)$ intersects every element of $\zP$.  If so, we return a set satisfying 1.  Thus, we may assume that there exists an element $P' \in \zP$ which is contained in $(D - C) \cap (Y-X)$.  Consider the separations $(C \cap X, D \cup Y)$ and $(C \cup X, D \cap Y)$.  The first is a separation separating $T'$ from an element of $\zP$; thus it must have order at least $t = |X \cap Y|$.  We conclude that the order of $(C \cup X, D \cap Y)$ must be at most the order of $(C, D)$.  If $V(X) \nsubseteq V(C)$, we get a contradiction to the tightness of $(C, D)$.  Thus, $V(X) \subseteq V(C)$.  It follows that $(C, D)$ is a separation of order $t+1$ by the tightness of $(X, Y)$. 

As the path $P$ is not contained in $C$, it follows that $\bar{P}$ must contain at least two vertices in $C \cap D$.  As $\bar{P}$ is disjoint from $\{x_1, \dots, x_t\}$ and $\{x_1, \dots, x_t\} \subseteq V(X) \subseteq V(C)$, there exists an index $i$ such that $x_i \in V(C) \setminus V(D)$.  Thus, $x_i'$ is also an element of $V(C)$.  We apply Lemma \ref{lem:tight1} to the separations $(X, Y)$, $(U_i, W_i)$ and $(C, D)$ and conclude that $V(C \cap U_i) = X$, however we have just seen that $x_i' \in V(C \cap U_i)$, a contradiction.

This contradiction shows that $\bar{T}$ is $\zP$-free, completing the proof of correctness for the algorithm.  The total runtime is $4^{t} n^{O(1)}$, as desired.
\end{proof}

We are now ready to proceed with the proof of Theorem \ref{thm:mainEP_revised}.

\restatetheorem{\restatemainEP*}
\begin{proof}
Let $G$, $T$, $H$, $k$, and $r$ be given.  Let $n = |V(G)|$, $m = |E(G)|$, and $m' = |E(H)|$.   Let $\zP$ be the set of truncated valid paths.  

Beginning with $T' = \emptyset$, we reiterate the algorithm from Lemma \ref{lem:alg1} up to $ 5^{10(k+r)}$ times to find one of the following:
\begin{enumerate}
\item a subset $Z$ of at most $4 \cdot  5^{20(k+r)}$ vertices intersecting every path in $G$ whose endpoints are $H$-adjacent;
\item a separation $(X, Y)$ of $G$ of order at most $2 \cdot  5^{10(k+r)}$ such that both $X$ and $Y$ contain an element of $\zP$;
\item a subset $T'$ such that $T' \subseteq T$, $|T'| = 2 \cdot 5^{10(k+r)}$, $H[T']$ contains a perfect matching, and $T'$ is $\zP$-free.
\end{enumerate}
At each iteration of the algorithm from Lemma \ref{lem:alg1}, note that $|T'| \le 2 \cdot 5^{10(k+r)} -2$, so that if we ever find the hitting set $Z$ in outcome 1 of Lemma \ref{lem:alg1}, $|Z| \le (2 \cdot 5^{10(k+r)} -2)(2 \cdot 5^{10(k+r)} +1)\le  4 \cdot  5^{20(k+r)}$, as desired.  

Given the runtime of the algorithm from Lemma \ref{lem:alg1}, we find one of the outcomes 1-3 above in time $4^{2\cdot 5^{10(k+r)}} n^{O(1)}\cdot 5^{10(k+r)}$.
If we find the hitting set in outcome 1, we return $Z$ and the algorithm terminates. Thus, we may assume we find either outcome 2 or 3.

Assume, as a case, we find a separation $(X, Y)$ of $G$ satisfying outcome 2.  We find valid paths $P_X$ and $P_Y$ such that $P_X - T$ (resp. $P_Y - T$) is contained in $X$ (resp. $Y$).  This can be done in time $O(n+m+m')$.  As the endpoints of $P_X$ and $P_Y$ have degree one, we may assume that $P_X \subseteq X$ and $P_Y \subseteq Y$.  Let $T_X = V(X - Y) \cap T$ and $T_Y = V(Y - X) \cap T$.  We find the induced subgraphs $H[T_X]$ and $H[T_Y]$ in time $O(m')$.  We recursively run the algorithm on $X-Y$, $H_X$, $T_X$, $k-1$ and on $Y-X$, $H_Y$, $T_Y$, $k-1$.   The runtime of the recursive calls is $4^{3 \cdot 5^{10((k-1)+r)}}[|V(X - Y)|^{O(1)} + |V(Y-X)|^{O(1)}]$. If we find $k-1$ valid paths in $X-Y$ or $Y-X$, we return these paths along with either the path $P_X$ or $P_Y$ and the algorithm terminates.  If we find sets $Z_X$ and $Z_Y$ hitting all the valid paths in the respective subgraphs, we return the set $Z_X \cup Z_Y \cup V(X \cap Y)$. Note that 
\begin{eqnarray*}
|Z_X \cup Z_Y \cup V(X \cap Y)| & \le &  2 \left( 4 \cdot 5^{20(k-1+r)}\right)+ 2 \cdot 5^{10(k+r)} \\ &\le & 4 \cdot 5^{20(k+r)}
\end{eqnarray*}
as desired.  Finally, if we find a subset of $T_X$ or $T_Y$ inducing a matching of size $r$, we return the subset to satisfy outcome 3.  

We conclude that we find the set $T'$ satisfying outcome 3 in the repeated iterations of the algorithm of Lemma \ref{lem:alg1}.  By Lemma \ref{lem:ramsey}, in polynomial time we can either find a subset of $T'$ inducing a matching of size $t$ in $H$, or find subsets $B_1, B_2 \subseteq T'$, $B_1 \cap B_2 = \emptyset$, and $|B_1| = |B_2| = k$ such that every vertex in $B_1$ is $H$-adjacent to every vertex in $B_2$.  If we find an induced matching of size $t$ in $H$, we return that subgraph; thus we may assume we have subsets $B_1$ and $B_2$ as above. 

We attempt to find $k$ disjoint paths linking $B_1$ and $B_2$.  If such paths exist, then we have found $k$ disjoint valid paths as desired.  Thus, we may assume there exists a separation of order at most $k-1$ separating the sets $B_1$ and $B_2$.  Assume we include all the vertices of $T' \setminus (B_1 \cup B_2)$ in the separator, and we conclude that there exists a separation $(X', Y')$ of order at most $|T' \setminus (B_1 \cup B_2)| + k-1 = |T'| - k - 1$ with $B_1 \subseteq V(X')$ and $B_2 \subseteq V(Y')$ and $T' \setminus (B_1 \cup B_2) \subseteq X' \cap Y'$.  Moreover, we can find the separation $(X', Y')$ in polynomial time.  We check in time $O(n+m+m')$ whether $X' - (B_1 \cup V(X' \cap Y'))$ or $Y' - (B_2 \cup V(X' \cap Y'))$ contain an element of $\zP$.  If not, we return $V(X' \cap Y') \cup B_1 \cup B_2$ as a set of at most $|T'| + k-1$ vertices intersecting all valid paths.  Otherwise, without loss of generality, assume $Y' - (B_2 \cup V(X' \cap Y'))$ contains an element of $\zP$.  For the moment, let $B_2$ also denote the subgraph with vertex set $B_2$ and no edges.  The separation $(X' \cup B_2, Y')$ is a separation of order at most $|T'|-1$ with $T' \subseteq V( X'\cup B_2)$ separating $T'$ from an element of $\zP$, contrary to our assumptions on $T'$.   
\end{proof}

\fi


\else
\input{approx_short.tex}
\fi

\section{Excluding induced matchings and skew bicliques: the exact FPT algorithm}
\label{sec:fptalg}

\renewcommand{\relemma}{lemma}

The goal of this section is to prove the algorithmic part of
Theorem~\ref{thm:main}: an FPT algorithm for \maxhhdisjoint if $\cH$
does not contain arbitrarily large induced matchings and skew
bicliques.  We state the algorithm in a robust way: even if the demand
graph $H$ contains large induced matchings and skew bicliques, the
algorithm works, but either returns a correct answer or returns a
large induced matching or a skew biclique of the demand graph $H$.

\begin{theorem}[Main algorithm]\label{thm:mainalg}
There is an algorithm that, given an instance $(G,T,H,k)$ of \maxdisjoint and an integer $r$, in time $f(k,r)\cdot n^{O(1)}$ either
\begin{itemize}
\item finds $k$ pairwise vertex-disjoint valid paths,
\item correctly states that there is no set of $k$ pairwise vertex-disjoint valid paths,
\item returns an induced matching of size $r$ in $H$, or
\item returns an induced skew biclique of size $r+r$ in $H$.
\end{itemize}
\end{theorem}
We do not estimate the function $f(k,r)$ of Theorem~\ref{thm:mainalg}
here, but as the algorithm eventually depends on the \disjointpath
algorithm (Theorem~\ref{thm:kdisjointpath}), it is a tower of some
number of exponentials.

Similarly to Section~\ref{sec:approx}, by attaching a new degree-1
vertex to every terminal and moving the endpoints of the demand edges
to these vertices, we may assume that $T$ is an independent set of
degree-1 vertices.  As an opening step, we invoke the algorithm of \iffull\ Theorem~\ref{thm:mainEP_revised} from Section~\ref{sec:approx}.\else
Theorem~\ref{thm:mainEP}.\fi\
 If it
returns a solution with $k$ pairwise disjoint valid paths, then we
are done. Otherwise, the algorithm returns a hitting set $Z$ of size
$2^{O(k+r)}$ that covers every valid path, that is, for any
connected component $C$ of $G-Z$, no two vertices of $V(C)$ are
adjacent in $H$.  As every terminal is degree-1, we may assume that
$Z$ is disjoint from $T$: if some terminal $t$ is in $Z$, then we may
replace it with its unique neighbor.  In this section, we assume that
such a set $Z$ is available and use the structural information given
by $Z$ to solve the problem.

If the number of terminals can be bounded by a function of $k$, then
we can enumerate every sequence $(s_1,t_1)$, $\dots$, $(s_k,t_k)$ of $k$ pairs of terminals such that $s_i$ and $t_i$ are adjacent in $H$ and invoke
the algorithm of Theorem~\ref{thm:kdisjointpath} for each such sequence. Therefore, our goal is to
reduce number of terminals to a constant depending only on $k$. The
main tool for this reduction is the notion of irrelevant terminals.

Given an instance $(G,T,H,k)$ of \maxdisjoint, we say that a terminal
$t\in T$ is {\em irrelevant} if $(G,T,H,k)$ is a yes-instance if and
only $(G,T\setminus \{t\},H,k)$. Note that, formally, if $(G,T,H,k)$
is a no-instance, then every terminal is irrelevant. In a
yes-instance, if there are more than $2k$ terminals, then some
terminal is surely irrelevant.  However, the main question is whether
we can identify provably irrelevant terminals in a reasonable running
time. The main technical result of the section is showing that if we
have a bounded-size hitting set $Z$ of the valid paths and there are
many terminals, then we can identify an irrelevant terminal in FPT
time. Therefore, we can remove that vertex from the set of terminals
and repeat the process until the number of terminals becomes bounded
by a constant depending only on $k$. We formulate the following result
in such a way that the algorithm either finds an irrelevant terminal,
a solution with $k$ disjoint paths, or one of the forbidden induced
subgraphs in $H$ (induced matchings or skew bicliques).

\begin{restatable}[Irrelevant terminal from a hitting set]{\relemma}{restateirrelevantmain}\label{lem:irrelevantmain}
  For every $k$, $r$ and $z$, there is a constant $\sm$ such that the
  following holds. Let $(G,T,H,k)$ be an instance of \maxdisjoint, let $r$ be an integer, and
  let $Z\subseteq V(G)\setminus T$ be a set of at most $z$ vertices such that $G-Z$ does not contain a valid path. If $|T|> \sm$, then in time $f(k,r,z)\cdot n^{O(1)}$, we can either
\begin{itemize}
\item find an irrelevant terminal $x\in T$,
\item a set of $k$ pairwise disjoint valid paths,
\item return an induced matching of size $r$ in $H$, or 
\item return an induced skew biclique of size $r+r$ in $H$.
\end{itemize}
\end{restatable}
Theorem~\ref{thm:mainalg} follows easily from Lemma~\ref{lem:irrelevantmain} and Theorem~\ref{thm:kdisjointpath}.

\iffull
\begin{proof}[Proof (of Theorem~\ref{thm:mainalg})]
  Let us modify first the instance such that $T$ is an independent set
  of degree-1 vertices in $G$. Let us invoke the algorithm of
  Theorem~\ref{thm:mainEP_revised}. If it returns a solution with $k$
  vertices or an induced matching of size $r$ in $H$, then we are
  done. Otherwise, we get a set $Z$ of vertices that covers every
  valid path. We may assume that $Z$ is disjoint from $T$, as we
  can replace any terminal in $Z$ with its unique neighbor: the
  resulting set still has the property that covers every valid
  paths.

  If $|T|>\sm$, then we invoke the algorithm of
  Lemma~\ref{lem:irrelevantmain}. If it returns an induced matching or
  a skew biclique in $H$, then we are done. Otherwise, if it returns
  an irrelevant terminal $v$, then we remove $v$ from the set of
  terminals, that is, we continue with the instance $(G,T\setminus
  \{v\},H-\{v\},k)$. We repeat this steps as long as $|T|>\sm$
  holds. If $|T|\le \sm$, then we enumerate every sequence
  $(s_1,t_1)$, $\dots$, $(s_k,t_k)$ of pairs of vertices from $T$ such
  that $s_i$ and $t_i$ are adjacent in $H$. There are at most
  $|T|^{2k}\le \sm^{2k}$ such sequences, which is number that can be
  bounded by a function of $k$, $z$, and $r$ only. For each such sequence, we
  use Theorem~\ref{thm:kdisjointpath} to find disjoint paths
  connecting these pairs of vertices. It is clear that the
  \maxdisjoint instance has a solution if and only if the algorithm
  of Theorem~\ref{thm:kdisjointpath} returns $k$ disjoint paths for at
  least one of these sequences.
\end{proof}
\fi

\iffull The proof of Lemma~\ref{lem:irrelevantmain} appears in Sections
\ref{sec:handling-cliques}--\ref{sec:boundcomponents}.
\fi
 Let us review
here the main ideas of the proof.

\textbf{Handling large cliques in $H$.}  As a first step, we show how
to find an irrelevant terminal given a large clique $K$ of $H$\iffull\ (Section~\ref{sec:handling-cliques})\fi. The special case of the disjoint paths problem when the demand pattern is a clique is a well-understood problem and we can use standard
polynomial-time algorithms to find $k$ disjoint paths with endpoints
in $K$. If there are $k$ such paths, then they form a solution of the
instance. Otherwise, a classical result of Gallai~\cite{MR0130188} shows that
there is a small set $S$ of vertices that cover every path with both
endpoints in $K$, or in other words, every connected component of
$G-S$ contains at most one vertex of $K$. Then we use this information to identify a 
vertex of $K$ that is an irrelevant terminal.

\begin{restatable}[Irrelevant terminal from a clique]{\relemma}{restatecliquereduce}\label{lem:cliquereduce}
There is polynomial-time algorithm that, given an instance of \maxdisjoint and a clique $K$ of $H$ having size $10k^2$, either
\begin{itemize}
\item returns a set of $k$ pairwise disjoint valid paths, or
\item returns an irrelevant terminal $t$.
\end{itemize}
\end{restatable}

\textbf{Separations and representative sets.}
Given a component $C$ of $G-Z$, we can define a separation $(A,B)$ with $C=V(A)\setminus V(B)$, which has the property that no two vertices of $A$ are adjacent in $H$. The main technical part of the proof is showing that if $A$ contains many terminals in such a separation, then we can find an irrelevant vertex.
\begin{restatable}[Irrelevant terminal from a separation]{\relemma}{restateirrelevantsep}\label{lem:irrelevantsep}
  For every $k$, $r$ and $z$, there is a constant $\smsep$ such that the
  following holds. Let $(A,B)$ be a separation of order at most $z$
with $|T\cap V(A)|>\smsep$, $V(A)\cap V(B)$ disjoint from $T$, and $H$ has no edge in $V(A)$. In time $f(k,r,z)\cdot n^{O(1)}$, we can either
\begin{itemize}
\item find an irrelevant terminal $x\in T\cap V(A)$,
\item return a set of $k$ pairwise disjoint valid paths,
\item return an induced matching of size $r$ in $H$, or 
\item return an induced skew biclique of size $r+r$ in $H$.
\end{itemize}
\end{restatable}

Given a solution and a separation $(A,B)$, let us focus on the part of
the solution inside $A$. An obvious and standard way of approaching
the problem would be to define an equivalence relation on these
partial solutions, where two partial solutions are equivalent if any
way of extending one of them to a full solution with edges in $B$ is
also a valid extension of the other partial solution. Let us enumerate
one partial solution from each equivalence class. If a terminal $t$ in
$A$ is not used by any of the enumerated partial solutions, then it is
irrelevant: if a partial solution is using $t$, then there is an
equivalent partial solution not using $t$, hence the solution can be
modified not to use $t$. If the number of equivalence classes is
bounded by a constant, then this gives a way of finding an irrelevant
terminal if the number of terminals in $A$ is larger than a constant.

Unfortunately, in our problem, the number of equivalence classes
cannot be bounded by any function of $k$ and the size of the
separation. A partial solution contains paths connecting a subset $T'$
of terminals in $A$ to the separator $V(A)\cap V(B)$, and the
equivalence class of the partial solution depends on what exactly this
subset is, as it determines which terminals can complete these paths
to valid paths of the solution. Therefore, the number of different
types a partial solution can have cannot be bounded by a function of
$k$ and the order $z$ of the separation only: it depends also on the
number $|T|$ of terminals and can be as large as $\Omega(|T|^z)$. For
example, let $G$ be a star with center $v$ and $2n$ leaves $a_1$,
$\dots$, $a_n$, $b_1$, $\dots$, $b_n$. Let $A=G[\{v,a_1,\dots,a_n\}]$
and $B=G[\{v,b_1,\dots,b_n\}]$. Now $(A,B)$ is a separation of order
1. Let $T=V(G)\setminus \{v\}$ and let $H$ be the matching with edges
$a_1b_1$, $\dots$, $a_nb_n$. Let $k=1$.  Now for $1\le i \le n$, the
partial solutions consisting of the single edge $a_iv$ are in
different equivalence classes: the edge $vb_i$ extends $a_iv$ to a
solution, but it does not extend $a_jv$ for any $j\neq i$. Therefore,
there are $n$ equivalence classes of partial solutions. By taking
disjoint unions of such stars, the reader may modify this example for
larger $k=z$ such that the number of terminals is $2kn$ and the number
equivalence classes is $n^k$.

One may hope that by excluding large induced matchings and large skew
bicliques, the number of equivalence classes can be bounded by a
constant.  Let us point out by a simple example that this is not the case. 
Let us modify the example in the previous paragraph such that $H$ is now a complete bipartite graph minus
the edges $a_ib_i$ for $1\le i \le n$. Observe that $H$ does not
contain large induced matchings and large induced skew bicliques.
Again, for $1\le i \le n$, the partial solutions consisting of
the single edge $a_iv$ are in different equivalence classes: the edge
$vb_i$ extends $a_jv$ to a solution for any $j\neq i$, but it {\em
  does not} extend $a_iv$. Therefore, again we have $n$ equivalence
classes of partial solutions.

We get around this problem using the idea of representative sets. We
show that, even though the valid partial solutions in a small
separation may form an unbounded number of equivalence classes, they
have a bounded-size subset that is representative in the sense that if
any partial solution can be extended to a correct solution, then one
of the partial solutions in the representative set can also be
extended to a correct solution. Continuing our example from the
previous paragraph, even though there are $n$ incomparable partial
solutions, there is a representative set consisting of only two
partial solutions, the edge $a_1v$ and the edge $a_2v$. Indeed, if a
solution contains the edge $vb_1$, then the part of the solution in
$A$ can be replaced with $a_2v$; if a solution contains the edge
$vb_i$ for $1<i \le n$, then the part of the solution in $B$ can be
replaced by $a_1v$. We show how to find a representative set of
partial solutions of bounded size. Then any terminal in $A$ that is
not used by any of these partial solutions can be considered to be
irrelevant. The bound and the algorithm relies heavily on the
assumption that the graph $H$ does not contain large cliques, large
induced matchings, skew bicliques graphs; or more precisely, the
algorithm either works correctly, or returns one such graph. If we
find a clique, then we can invoke Lemma~\ref{lem:cliquereduce}. By the
specification of Lemma~\ref{lem:irrelevantmain}, the induced matchings
or skew biclique can be returned.  The concept of representative sets
has been used in the design of FPT algorithms
\cite{MR87a:05097,DBLP:journals/corr/FominLPS14,DBLP:conf/soda/FominLS14,DBLP:journals/tcs/Marx09,DBLP:journals/tcs/Marx06a},
but our application does not follow from any of the earlier technical
statements; in particular, the fact that this approach works precisely
when there are no larges cliques, bicliques, or matchings is quite
specific to our problem.

On a high level, the proof of Lemma~\ref{lem:irrelevantsep} goes the
following way. Consider those paths of the solution that cross the
separator and have one endpoint in $A$ and one in $B$. The endpoints
in $A$ form a vector $\mathbf{a}$ and the endpoints in $B$ form a
vector $\mathbf{b}$. These two vectors are compatible in the sense
that the $j$-th coordinate of $\mathbf{a}$ is adjacent in $H$ with the $j$-th
coordinate of $\mathbf{b}$. The partial solution connects the vertices
in $\mathbf{a}$ to the separator $V(A)\cap V(B)$. If we want to
replace the partial solution with another partial solution that connects a different set $\mathbf{a}'$ of vertices to the separator, then we
have to make sure that the new vector $\mathbf{a}'$ is also compatible with the vector $\mathbf{b}$. Therefore,
if we classify the partial solution according to the vector of
terminals connected to the separator, then we have to find a
representative subset of these vectors in the sense that if some
vector $\mathbf{a}$ is compatible with some vector $\mathbf{b}$, then
the representative subset also contains a vector $\mathbf{a}'$
compatible with $\mathbf{b}$. \iffull In
Section~\ref{sec:representative-sets-vectors},
\else
First\fi\ we consider this
abstract problem on vectors, and show (assuming that $H$ has no large
induced matching or induced skew biclique) how we can find a
bounded-size representative set of vectors. Of course, our problem is
more complicated than just matching these vectors, for example, a path
in the solution can cross the separator several times. \iffull 
In Section~\ref{sec:repr-sets-disj},\else Then\fi\ we
address these issues by classifying the partial solutions into a bounded
number of types according (mostly) to what happens at the separator.
\iffull We conclude the proof of Lemma~\ref{lem:irrelevantsep} in
Section~\ref{sec:repr-sets-disj}.\fi

\textbf{Reducing the number of components.}  Finally, after we reduced
the number of terminals in each component of $G-Z$ with repeated
applications of Lemma~\ref{lem:irrelevantsep}, our goal is to reduce
the number of components of $G-Z$ that contain terminals. \iffull In
Section~\ref{sec:boundcomponents}, we\else We\fi\ show that this can be done quite
easily by a simple marking procedure. The proof relies on the fact
that the number of terminals is bounded in each component. Thus it does
not seem to be easy to do the reduction of the number of components
before the reduction of the number terminals in the components.
\iffull
\subsection{Handling cliques}
\label{sec:handling-cliques}

In this section, we discuss how to find an irrelevant vertex if we
have a large clique in the demand graph $H$
(Lemma~\ref{lem:cliquereduce} above).  The reason why we are treating
this special case separately is that a combinatorial argument of the
following section (Lemma~\ref{lem:noskew}) works only if we can assume
that there are no large induced matchings, skew bicliques, and cliques
in the demand graph $H$. By the specification of
Theorem~\ref{thm:mainalg}, if we encounter large induced matchings or
skew bicliques, then we may stop, but there is no reason why large
cliques cannot appear in the demand graph $H$. Therefore, we need some
argument to handle large cliques, and this is what we provide in this
section.  Note that even if we have a procedure handling large
cliques, we cannot say the we apply it exhaustively on every
sufficiently large clique of $H$ and after that it can be assumed that $H$ has
no large cliques: finding a clique of size $k$ is
\classWone-hard. Instead, what we do is whenever the algorithm
described in the following section fails because it finds a a large
clique, then we invoke this procedure.

The following result was proved by Gallai~\cite{MR0130188} in a combinatorial form, the algorithmic version is folklore:
\begin{theorem}[Gallai~\cite{MR0130188}]\label{thm:mway}
  Given an undirected graph $G$, a set $A\subseteq V(G)$ of vertices, and an integer $k$, we can find in polynomial time either
\begin{itemize}
\item a set of $k$ pairwise vertex-disjoint paths with endpoints in $A$, or
\item a set $S$ of at most $2k-2$ vertices such that every component of $G- S$ contains at most one vertex of $A\setminus S$.
\end{itemize}
\end{theorem}

Using Theorem~\ref{thm:mway} on a sufficiently large clique $K$ of $H$, we may either find $k$ valid paths forming a solution or we can identify a terminal of $K$ that can be always avoided in a solution.

\restatelemma{\restatecliquereduce*}
\begin{proof}
  By Theorem~\ref{thm:mway} applied to graph $G$, vertices $K$, and
  integer $k$, we can find in polynomial time either $k$ pairwise
  vertex-disjoint paths with endpoints in $K$, or a set $S$ of size at
  most $2k-2$ such that every component of $G- S$ contains at
  most one vertex of $K$. In the former case, we return this set of
  $k$ paths as a valid solution. In the later case, let $S'\subseteq
  S$ contain a vertex $v\in S$ if there are at least $5k+1$ components
  of $G- S$ that are adjacent to $v$ and intersect $K$ (in exactly one terminal).  This means
  that there are at most $5k|S\setminus S'|\le 5k|S|\le 5k\cdot (2k-2)=10(k^2-k)$ components $C$
  of $G- S$ with $C\cap K\neq\emptyset$ and $N(C)\not\subseteq
  S'$. As $|K|\ge 10k^2$ and hence $|K-S|\ge 10k^2-(2k-2)> 10(k^2-k)+k$, there exists $k$ components $C_1$, $\dots$, $C_k$ with
  $|C_i\cap K|=1$ and $N(C_i)\subseteq S'$.  If each of these
  $k$ components fully contains a valid path, then picking a valid
  path from each of them gives a solution that we can
  return. Otherwise, there is a component $C$ of $G- S$ with
  $C\cap K=\{t\}$, $N(S)\subseteq S'$, and not containing a valid
  path.

  We claim that removing $t$ from the set of terminals gives an
  equivalent instance. That is, we show that any solution containing a
  path $P$ with endpoint $t$ can be modified in such a way that it
  does not use $t$. By the choice of $C$, there is no valid path in
  $C$, hence we know that $P$ is not contained fully in $C$. Let $v$
  be the vertex of $N(C)\subseteq S'$ that is closest to $t$ on
  $P$. As $v\in S'$, there are at least $5k+1$ components of
  $G\setminus S'$ intersecting $K$ and adjacent to $v$. At most $k-1$
  of them can contain fully a path of the solution (different from
  $P$) and at most $2|S|\le 4k$ of them can contain a path going
  intersecting $|S|$ (observe that a path containing $x$ vertices of
  $|S|$ can intersect at most $x+1\le 2x$ components). Therefore,
  there are two such a components $C^1, C^2$ disjoint from every path
  of the solution; let $C^1\cap K=\{t_1\}$ and $C^2\cap
  K=\{t_2\}$. Now the path $P$ can be replaced by a path connecting
  $t_1$ and $t_2$ via $v$. This proves the claim that removing $t$
  from the set of terminals gives an equivalent instance.
\end{proof}

\subsection{Representative sets for vectors of vertices}
\label{sec:representative-sets-vectors}

In this section, we prove a statement about representative sets in an abstract setting of compatible vectors (Lemma~\ref{lem:vectorrepbound} below). In Section~\ref{sec:repr-sets-disj}, we use this result to prove a bound on the size of representative sets of partial solutions, which will allow us to find irrelevant terminals if a component of $G-Z$ contains too many terminals. 
\begin{definition}\label{def:repvector}
  Let $H$ be an undirected graph and let $d$ be a positive integer. We
  say that two $d$-tuples $(a_1,\dots,a_d),(b_1,\dots,b_d)\in V(H)^d$
  are {\em compatible} if $a_i$ and $b_i$ are adjacent in $H$ for every $1\le
  i \le d$. Let $\R\subseteq V(H)^d$ be a set of $d$-tuples. We say
  that $\R'\subseteq \R$ is a {\em representative} subset of $\R$ if
  for every compatible pair $\mathbf{a}\in \R$ and $\mathbf{b}\in
  V(H)^d$, there is an $\mathbf{a}'\in \R'$ such that $\mathbf{a}'$ and
$\mathbf{b}$ are compatible.
\end{definition}

Note that we do not require that the coordinates of a vector
$(a_1,\dots,a_d)$ be all distinct, and $(a_1,\dots,a_d)$ and
$(b_1,\dots, b_d)$ can be compatible even if $a_i=b_j$ for some $i\neq
j$ (but $a_i=b_i$ is clearly impossible, as no vertex of $H$ is
adjacent to itself).

We need the following simple Ramsey argument, whose proof is very
similar to the proof of Lemma~\ref{lem:ramsey} in
Section~\ref{sec:approx}.
\begin{lemma}\label{lem:noskew}
  Let $r$ and $n$ be positive integers with $n\ge
  4^{4r}$.
Let $H$ be a graph and $a_1$, $\dots$, $a_n$, $b_1$, $\dots$, $b_n$ be distinct vertices such that
\begin{itemize}
\item $a_i$ and $b_i$ are adjacent for $1\le i \le n$, and
\item $a_i$ and $b_j$ are not adjacent for $1\le i < j \le n$.
\end{itemize}
Then in polynomial time we can find either
\begin{itemize}
\item an induced matching of size $r$ in $H$,
\item an induced skew biclique on $r+r$ vertices in $H$, or
\item a clique of size $r$ in $H$.
\end{itemize}
\end{lemma}

\begin{proof}
By Lemma \ref{lem:multiramsey}, every clique with $4^{4r}$ vertices such that the edges are colored by one of four colors contains a clique subgraph of size $r$ where all the edges are the same color.    

 Consider the clique on $n$ vertices, with the vertices labeled $1, \dots, n$.  We define a 4-coloring of the edges as follows.  For an edge of the clique $ij$ with $i < j$, we color the edge:
\begin{enumerate}
\item color 1 if no edge of $H$ has one end in $\{a_i, b_i\}$ and one end in $\{a_j, b_j\}$, 
\item color 2 if $a_i$ is adjacent $a_j$, 
\item color 3 if $b_i$ is adjacent $b_j$ and $a_i \nsim a_j$, 
\item color 4 if $b_i$ is adjacent $a_j$ and $a_i \nsim a_j$, $b_i \nsim b_j$
\end{enumerate}
where all adjacencies are in the graph $H$.
Note that this covers every possibility, as we know by assumption that $a_i\nsim b_j$ for $i<j$. Therefore, this defines a 4-coloring of the edges of the clique.  By our choice of $n$, there exists a subset of vertices of size $r$ inducing a monochromatic subclique, and we can identify it in polynomial time.  Without loss of generality, we may assume that the vertices $1 \le i \le r$ of the clique induce such a monochromatic clique.  If the subclique has color 1, then $H$ contains an induced matching of size $r$.  If the monochromatic clique has color 2 or 3, then $H$ contains a clique subgraph of size $r$.  Finally, if the subclique has color 4, then the graph contains an induced skew biclique on $r+r$ vertices.
\end{proof}

The following lemma states that (assuming there is no large induced
matching, skew biclique, or clique in $H$) every set of vectors has a
bounded-size representative subset.
\begin{lemma}[Representative set bound]\label{lem:vectorrepbound0}
  Let $H$ be an undirected graph, $r$ and $d$ positive integers, and
  $\R\subseteq V(H)^d$ a set of $d$-tuples. Suppose that there is no
  induced matching of size $r$, induced skew biclique of size $r+r$,
  or clique of size $r$ in $H$.  Then there is a representative subset
  $\R'\subseteq \R$ of size at most $\Rmvec:=(d+1)^{d(4^{4r})}$.
\end{lemma}
We prove an algorithmic version of
Lemma~\ref{lem:vectorrepbound0}. The straightforward algorithmic
statement would be to say that, given a set $\R$ of vectors, a
bounded-size representative set can be found. However, we would like
to find small representative sets efficiently also for large,
implicitly given sets $\R$ that would be too time consuming to
enumerate explicitly. Therefore, we state the algorithmic version of
Lemma~\ref{lem:vectorrepbound0} in a way that $\R$ is given by a query
procedure that, given sets $A_1,\dots,A_d\subseteq V(H)$, returns a
vector $\mathbf{a}\in (A_1\times\dots \times A_d)\cap \R$, if such a
vector exists. 

\begin{lemma}[Representative set bound, algorithmic version]\label{lem:vectorrepbound}
  Let $H$ be an undirected graph, $r$ and $d$ positive integers, and
  $\R\subseteq V(H)^d$ a set of $d$-tuples.  Suppose that the set $\R$
  is given via a query procedure that, given sets $A_1, \dots,
  A_d\subseteq V(H)$, returns an $\mathbf{a}\in
  (A_1\times \dots \times A_d)\cap \R$, or states that no such vector
  $\mathbf{a}$ exists.  There is an algorithm whose running time is
  polynomial in $n$, in $\Rmvec:=(d+1)^{d(4^{4r})}$, and in the
  running time of the query procedure, and finds either
\begin{itemize}
\item a representative subset $\R'\subseteq \R$ of size at most $\Rmvec$,
\item an induced matching of size $r$ in $H$,
\item an induced skew biclique on $r+r$ vertices in $H$, or
\item a clique of size $r$ in $H$.
\end{itemize}
\end{lemma}
\begin{proof}
  The algorithm builds a rooted tree where each node is either empty
  or contains a compatible pair $(\mathbf{a},\mathbf{b})$ with
  $\mathbf{a}\in \R$ and $\mathbf{b}\in V(H)^d$. Empty nodes have no
  children and each nonempty node has exactly $d$ ordered children.
  Initially, we start with a tree consisting of a single empty node.

  For a vector $\mathbf{b'}\in V(H)^d$, we define the following search
  procedure on the tree.  We start the procedure at the root node.  If
  the current node is empty, then we say that the procedure fails at
  this empty node.  Otherwise, let $(\mathbf{a},\mathbf{b})$ be the current
  node. If $\mathbf{a}$ and $\mathbf{b}'$ are compatible, then we
  declare the search to be successful. Otherwise, let $1\le j \le d$
  be the first coordinate such that the $j$-th coordinates of
  $\mathbf{a}$ and $\mathbf{b}'$ are not adjacent. Then we continue
  the search at the $j$-th child of the current node.

  Given an empty node $u$ of the tree, we show how to check whether
  there are $d$-tuples $\mathbf{a}'=(a'_1,\dots, a'_d)\in \R$ and
  $\mathbf{b}'=(b'_1,\dots,b'_d)\in V(H)^d$ such that $\mathbf{a}'$
  and $\mathbf{b}'$ are compatible and $\mathbf{b}'$ fails at
  $u$. Consider the path from the root of the tree to the empty node
  $u$. Let $(\mathbf{a},\mathbf{b})$ be a nonempty node on this path
  such that the path continues with the $j$-th child of this nonempty
  node.  Then the $j$-th coordinate of $\mathbf{a}$ is {\em not}
  adjacent to the $j$-th coordinate of $\mathbf{b}'$, while for every
  $1\le j' < j$, the $j'$-th coordinate of $\mathbf{a}$ is adjacent to
  the $j'$-th coordinate of $\mathbf{b}'$. These requirements together
  give a subset $B_j\subseteq V(H)$ of potential values for the
  $j$-th coordinate of $\mathbf{b}'$. Now a vector $\mathbf{b'}\in
  V(H)^d$ fails at $u$ if and only if $\mathbf{b'}\in
  B_1\times\dots\times B_d$.  Therefore, we need to find a vector
  $\mathbf{a}'\in \R$ that is compatible with at least one vector in
  $B_1\times\dots\times B_d$. Let $A_j$ contain every vertex of $H$
  that has at least one neighbor in $B_j$. Observe that a vector
  $\mathbf{a}'\in V(H)^d$ is compatible with at least one vector in
  $B_1\times \dots \times B_d$ if and only if $\mathbf{a}'\in
  A_1\times\dots\times A_d$. Therefore, we can use the query procedure
  to check the existence of such a vector
  $\mathbf{a}'=(a'_1,\dots,a'_d)$ and then we can construct
  $\mathbf{b}'=(b'_1,\dots,b'_d)\in B_1\times \dots\times B_d$ by
  letting $b'_j$ be an arbitrary neighbor of $a'_j$ in $B_j$.

  We consider every empty node $u$ (in arbitrary order) and use the
  method described in the previous paragraph to find an
  $\mathbf{a}'\in \R$ and a $d$-tuple $\mathbf{b}'$ compatible with
  $\mathbf{a}'$ that fails at $u$. If there is such a pair $(\mathbf{a}',\mathbf{b}')$,
  then we replace $u$ with $(\mathbf{a}',\mathbf{b}')$ and add $d$
  empty children to this node. We repeat this step until no such
  $\mathbf{b}'$ can be found for any empty node $u$.  At this point,
  let us define the set $\R'=\{\mathbf{a} \mid
  \text{$(\mathbf{a},\mathbf{b})$ appears in a nonempty node}\}$, that
  is, $\R'$ contains the first part of every pair appearing in
  the tree. Clearly, we have $\R'\subseteq \R$ from the way new
  nonempty nodes are introduced into the tree. Moreover, we claim that
  $\R'$ is a representative subset of $\R$. Indeed, for every
  adjacent pair $\mathbf{a}'\in \R$ and $\mathbf{b}'\in V(H)^d$, the
  search procedure for $\mathbf{b}'$ cannot fail at any empty node $u$
  (otherwise we would have extended the tree at $u$) and therefore the
  tree contains a pair $(\mathbf{a},\mathbf{b})$ such that
  $\mathbf{a}\in \R'$ is compatible with $\mathbf{b}'$.

  We prove that if the height $h$ of the tree reaches $d(4^{4r})$,
  then we can find an induced matching, a skew biclique, or a clique
  of the specified size and we can stop the algorithm. Otherwise, if
  the algorithm terminates without stopping this way, then every path
  from the root to a leaf contains less than $d(4^{4r})$ nonempty
  nodes, and hence the number of nonempty nodes is at most
  $\sum_{i=0}^{d(4^{4r})-1}d^i\le 
  (d+1)^{d(4^{4r})}=\Rmvec$. Therefore, as showed in the previous
  paragraph, we obtain a representative subset $\R'$ of size at most
  $\Rmvec$.

  Consider a path from the root to a leaf with at least $d(4^{4r})$
  nonempty nodes. Then there is a $1\le j \le d$ such that it is true
  for at least $n=4^{4r}$ nodes on the path that the path
  continues with the $j$-th child of the node. Let
  $(\mathbf{a}_1,\mathbf{b}_1)$, $\dots$,
  $(\mathbf{a}_n,\mathbf{b}_n)$ be $n$ such nodes, ordered as they
  appear on the path from the root to the leaf. Let $a_i$ and $b_i$ be
  the $j$-th coordinate of $\mathbf{a}_i$ and $\mathbf{b}_i$,
  respectively. As the pair $(\mathbf{a}_i,\mathbf{b}_i)$ is
  compatible, we have that $a_i$ and $b_i$ are adjacent. Furthermore,
  consider the execution of the search procedure when $\mathbf{b}_i$
  failed and the node $(\mathbf{a}_i,\mathbf{b}_i)$ was added to the
  tree. Note that the tree is extended only by replacing leaf nodes,
  thus the ancestors of $(\mathbf{a}_i,\mathbf{b}_i)$ did not change
  after they were added to the tree. Therefore, for every $1\le i' <
  i$, the search procedure for $\mathbf{b}_{i}$ encountered the node
  $(\mathbf{a}_{i'},\mathbf{b}_{i'})$ and then continued the search
  with the $j$-th child of this node. This means that the $j$-th
  coordinate of $\mathbf{b}_{i}$ is not adjacent to the $j$-th
  coordinate of $\mathbf{a}_{i'}$. That is, we get that $a_{i'}$ is
  not adjacent to $b_{i}$ for every $1\le i' < i \le n$. Therefore,
  the conditions of Lemma~\ref{lem:noskew} hold, and we can use it to
  return an induced matching, a skew biclique, or a clique.
\end{proof}

\subsection{Representative sets for disjoint paths}\label{sec:repr-sets-disj}
We can describe a solution as a subgraph $P$ of $G$ that is the union of $k$ pairwise-disjoint valid paths. A {\em partial solution} is any subgraph of $G$ that is the union of disjoint paths (possibly more than $k$ or possibly with endpoints not in $T$).
Given a solution $P$ and a separation $(A,B)$ of $G$,
the {\em partial solution} of $P$ at $(A,B)$ is the
subgraph $\Pi$ of $P$ induced by $V(A)$.
To define representative sets of partial solutions, we need to define first what it means to replace a partial solution with another:
\begin{definition}
Let $P$ be a solution, let $(A,B)$ be a separation of $G$, and let $\Pi$ be a partial solution at $(A,B)$. We say that $\Pi$ is {\em replacable} at $(A,B)$
in $P$ if the subgraph $P'=(P-E(G[V(A)]))\cup \Pi$ is a valid solution. 
In this case, we say that $P'$ is obtained by {\em replacing} $\Pi$ into $P$ at $(A,B)$.
\end{definition}

\begin{definition}
  Let $\R$ be a set of partial solutions at $(A,B)$. We say that $\R$ is
  {\em representative} if for every solution $P$, there is a $\Pi\in
  \R$ that is replacable into $P$ at $(A,B)$. We say that a subset
  $\R'\subseteq \R$ {\em represents} $\R$ if for every solution $P$
  whose partial solution at $(A,B)$ is in $\R$, there is a $\Pi\in
  \R'$ that is replacable into $P$ at $(A,B)$.
\end{definition}

The main result of the section is the following:
\begin{lemma}[Irrelevant terminal or clique from a separation]\label{lem:repbound}
  For every $k$ and $z$, there is a constant $\Rm$ such that the
  following holds. Let $(A,B)$ be a separation of order $z$ such that
  $V(A)\cap V(B)$ is disjoint from $T$ and $H$ has no edge in $V(A)$.
  Let $\R$ contain the partial solution at $(A,B)$ for every solution. In time $f(r,z)\cdot n^{O(1)}$, we can either
\begin{itemize}
\item find a representative set $\R'\subseteq \R$ of partial solutions at $(A,B)$ with $|\R'|\le \Rm$,
\item return an induced matching of size $r$ in $H$, 
\item return an induced skew biclique on $r+r$ in $H$, or
\item return a clique of $r$ in $H$.
\end{itemize}
\end{lemma}
\begin{proof}
  Let $S=V(A)\cap V(B)$. Let $P$ be a solution and let $\Pi$ be the
  partial solution of $P$ at $(A,B)$.  As $H$ has no edge in $V(A)$,
  every path of the partial solution contains a vertex of $S$, hence
  there are at most $z$ paths in the partial solution.  Each path $P$
  can be classified into exactly one of the following three classes (recall that $S\cap T=\emptyset$); see Figure~\ref{fig:partial}:
\begin{enumerate}[label=(C\arabic*),start=0]
\item \label{t:0} $P$ consists of single vertex of $S$.
\item \label{t:1} $P$ has length at least one and has one endpoint in $T$ and one endpoint in $S$.
\item \label{t:2} $P$ has length at least one and has both endpoints in $S$.
\end{enumerate}
\begin{figure}
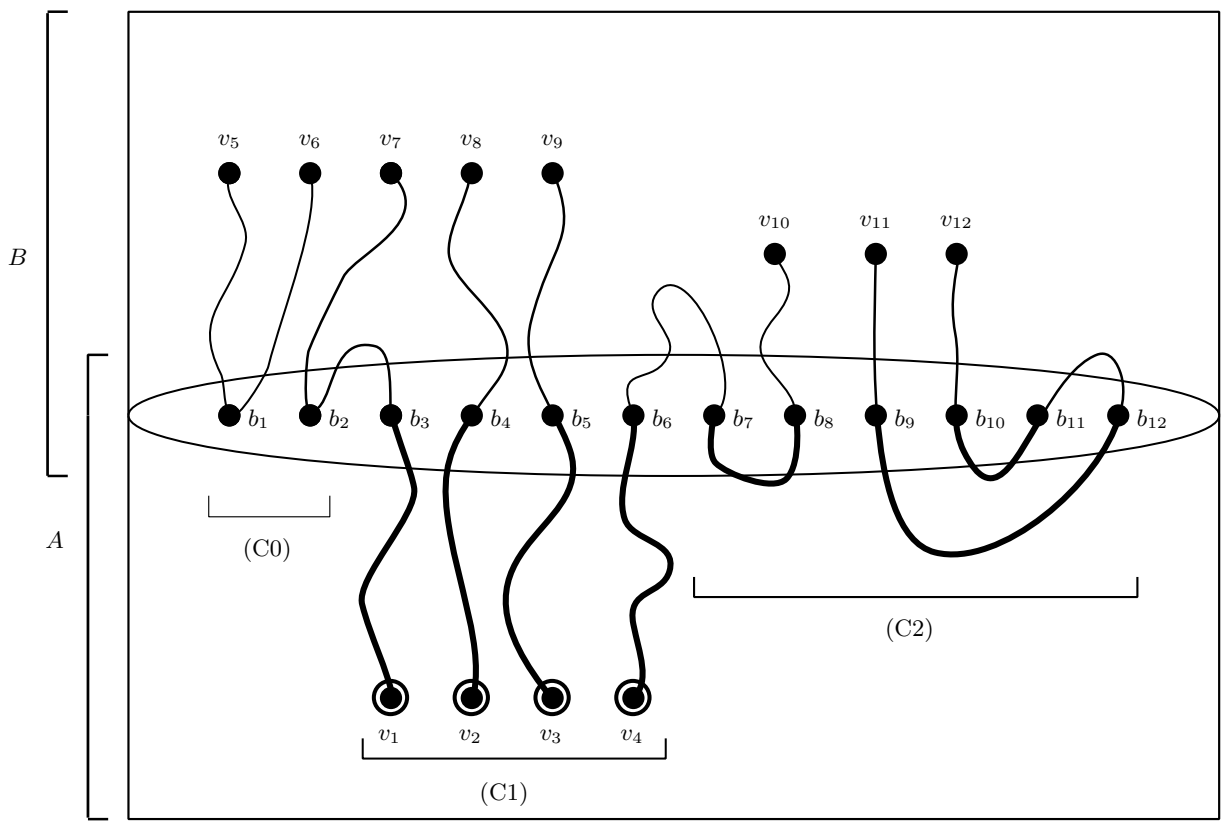

\begin{center}
{\footnotesize \svg{\linewidth}{partial}}
\caption{A partial solution at $(A,B)$
  (Lemma~\ref{lem:repbound}).  Set $S_0=\{b_1,b_2\}$ contains the vertices of the two paths of type \ref{t:0}. There are four paths of class
  \ref{t:1}, connecting $\{v_1,v_2,v_3,v_4\}$ to
  $J=\{b_3,b_4,b_5,b_6\}$. The three paths of type \ref{t:2} define the matching $M=\{b_7b_8, b_9b_{12},b_{10}b_{11}\}$. Assuming the ordering
  $(b_{3},b_{4},b_{4},b_{6})$ of $J$, the inner vector is
  $(v_1,v_{2},v_{3},v_4)$ and the outer vector is
  $(v_{7},v_{8},v_9,v_{10})$.}\label{fig:partial}
\end{center}
\end{figure}

The paths of class~\ref{t:2} define a (not necessarily perfect) 
matching 
$M$ of $S$ the obvious way. We define the {\em join vertex} of a path
$P$ of class \ref{t:1} to be its endpoint in $S$. 
Let $J\subseteq S$ be the join vertices of
the paths of class~\ref{t:1}. We
define the {\em type} of a partial solution $\Pi_A$ to be the
triple $\tau=(S_0,J,M)$, where
\begin{itemize}
\item $S_0\subseteq S$ is the set of vertices used by paths of class~\ref{t:0}.
\item $J\subseteq S$ is the set of join vertices of paths of class~\ref{t:1}, 
\item $M$ is the matching of $S$ defined above based on the paths of class~\ref{t:2}.
\end{itemize}
Note that the number of types is at most $T:=2^{z}\cdot 2^{z}\cdot z^z$. Let $\R_\tau\subseteq \R$ contain every partial
solution of type $\tau$. For every type $\tau$, we construct a
representative subset $\R'_\tau\subseteq \R_\tau$. It is clear
that the union $\R'$ of $\R'_\tau$ for every type $\tau$ is
representative subset of $\R$.

We construct $\R'_\tau$ for a type $\tau=(S_0,J,M)$ the following way. Let us fix an ordering of $J=(v_1,\dots, v_d)$ (note that $d\le z$). For a partial
solution $\Pi$ of type $\tau$, let $P_j$ be the path of
class~\ref{t:1} whose join vertex is $v_j$. Let $a_j$ be the other endpoint of
$P_j$.  We
define the $d$-tuple $\mathbf{a}=(a_1,\dots,a_d)\in V(H)^d$ as the
{\em inner vector} of the partial solution $\Pi$.
\begin{figure}
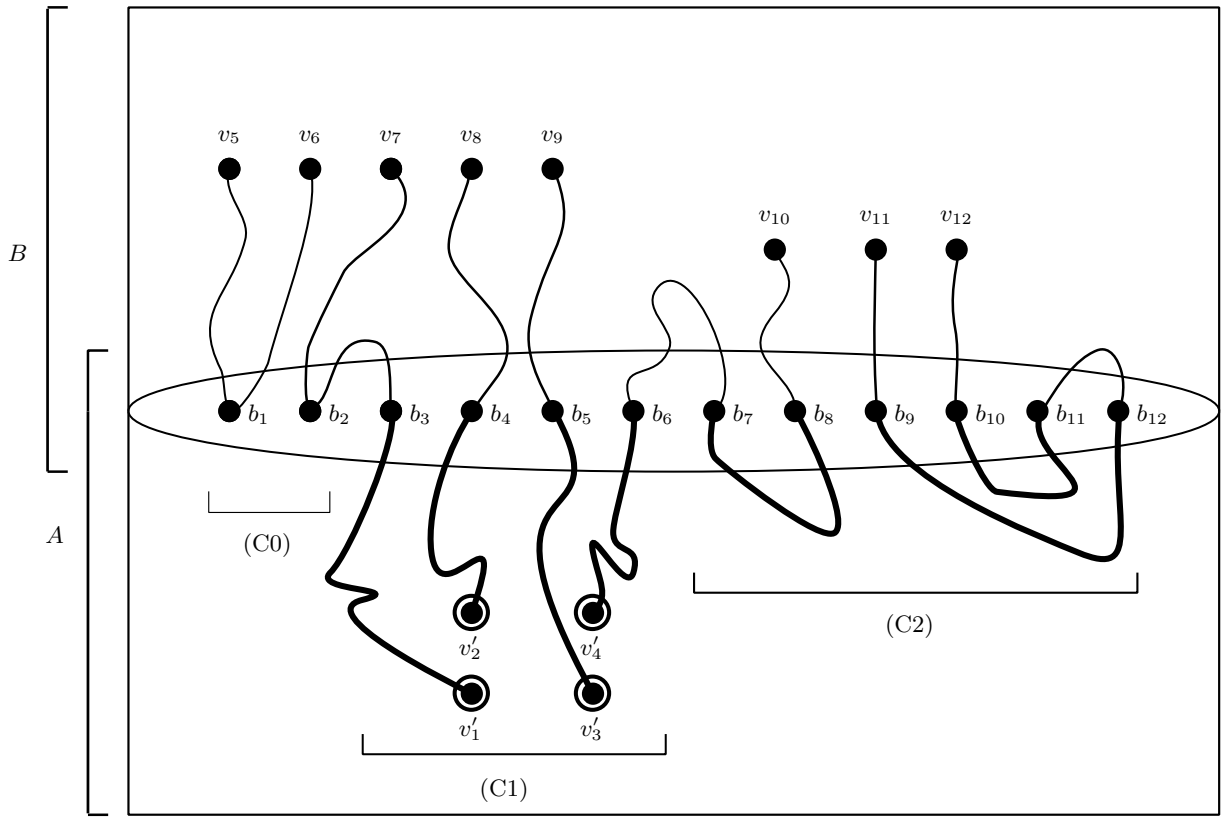

\begin{center}
{\footnotesize \svg{\linewidth}{partial2}}
\caption{A partial solution having the same type as
  the partial solution in Figure~\ref{fig:partial}, and replacing it at
  $(A,B)$. The inner vector is $(v'_1,v'_{2},v'_{3},v'_4)$.
  Assuming $H$ contains the edges $v'_1v_7$, $v'_2v_8$, $v'_3v_9$,
  $v'_4v_{10}$, the
  result of the replacement is a valid solution.
}\label{fig:partial2}
\end{center}
\end{figure}

Let $\bR_\tau$ be the inner vectors of the partial solutions in
$\R_\tau$.  We would like to invoke Lemma~\ref{lem:vectorrepbound} on
the set $\bR_\tau$.  For this purpose, we need to implement the query
procedure. We need to test the existence of a partial solution of type
$\tau$ whose inner vector is in $A_1\times\dots\times A_d$. We reduce
this question to solving an instance of the $k$-disjoint paths
problem.  As we have observed earlier, each partial solution of type
$\tau$ consist of a set of at most $z$ vertex-disjoint paths.  Let us
start with the graph $A-S_0$: we remove the set $S_0$, as it is
reserved for paths of class~\ref{t:0}.  For every pair $(v_1,v_2)$ in
the matching $M$, we introduce a corresponding pair in the constructed
\disjointpath instance: the paths of the solution connecting these
pairs will correspond to the requested paths of class~\ref{t:2} in the
partial solution.  To handle paths of type \ref{t:1}, let us introduce
a vertex $s_j$ adjacent to every vertex of $A_j$ for every $1\le j \le
d$. Then we specify the pairs $(s_j,v_j)$ for every $1\le j \le d$
(recall that the $v_j$'s are the vertices of $J$).  Let us use the
algorithm of Theorem~\ref{thm:kdisjointpath} to find vertex-disjoint
paths with the specified endpoints. If such a collection of disjoint
paths exist, then we obtain, after removing the vertices $s_1$, $\dots$, $s_d$, a
set of disjoint paths in $A$.  These paths form a partial solution of
type $\tau$ whose inner vector is in $A_1\times\dots \times A_d$,
hence the query procedure can return this partial solution. Conversely,
if there exists a partial solution $\Pi$ of type $\tau$ having inner
vector in $A_1\times \dots \times A_d$, then it gives a solution for
the constructed instance of \textsc{Disjoint Paths}.  This implies
that the algorithm of Theorem~\ref{thm:kdisjointpath} finds a solution
for this instance of \textsc{Disjoint Paths}, resulting in a
partial solution $\Pi'$ of type $\tau$ and inner vector in $A_1\times
\dots \times A_d$.

Using the query procedure described in the previous paragraph, we may
invoke Lemma~\ref{lem:vectorrepbound} on the set $\bR_\tau$.  If we
get an induced matching, induced skew biclique, or a clique,
then we are done. Otherwise, we get a representative subset
$\bR'_\tau$ of $\bR_\tau$ having size at most
$\Rmvec$. Note that each vector $\mathbf{a}$ introduced
into $\bR'_\tau$ was returned by the query procedure, which means that
the query procedure found a partial solution of type $\tau$ and inner
vector $\mathbf{a}$; let $\R'_\tau\subseteq \R_\tau$ contain every
such partial solution. Finally, we construct the set $\R'$ as the
union of $\R'_\tau$ for every type $\tau$; as both the number of types
and the size of each $\R'_\tau$ can be bounded by a function of $r$ and $z$ only, the size of $\R'$ can be bounded by a constant
$\Rm$ depending only on $r$ and $z$.

We claim that $\R'$ is also a representative set of partial solutions
at $(A,B)$. Let $P$ be a solution and let $\Pi\in \R$ be its partial solution at $(A,B)$. Suppose that $\Pi$ has type $\tau=(S_0,M,J)$ and let
$\mathbf{a}$ be the inner vector of $\Pi$. Recall that we
fixed an ordering $(v_1,\dots,v_j)$ of $J$, there is a path $P_j$
of type \ref{t:1} with endpoints $a_j$ and $v_j$ for every $1\le j
\le d$, and the inner vector is $(a_1,\dots, a_j)$.  For every $1\le j
\le d$, let $b_j$ be the other endpoint of the path of $a_j$ in the solution $P$. We define $\mathbf{b}=(b_1,\dots,b_d)\in V(H)^d$ as
the {\em outer vector} of the partial solution $\Pi$ in
$P$. Observe that the inner vector $\mathbf{a}$ and the outer
vector $\mathbf{b}$ are compatible.  As $\mathbf{a}\in \bR_\tau$ and
$\bR'_\tau$ is a representative subset of $\bR_\tau$, there is an
$\mathbf{a}'\in \bR_\tau$ that is also compatible with $\mathbf{b}$. Thus
there is a partial solution $\Pi'\in \R'_\tau\subseteq
\R'$ having inner vector $\mathbf{a}'$.

We claim that replacing $\Pi'$ at $(A,B)$ in $P$
gives a valid solution $P'$.  If a path
$P$ of $Q$ has both endpoints outside $V(A)$, then there is a
corresponding valid path after the modification: as the two partial
solutions have the same type, the set $S_0$ and the matching $M$ are the same in both
of them. Therefore, whenever $P$ has an $x-y$ subpath in $A$ for some
$x,y\in S$, then this subpath is a path of class \ref{t:0} or \ref{t:2} in $\Pi$,
hence there is a path with the same endpoints in
$\Pi'$.
If a path of $Q$ has one endpoint in $V(A)$, then the other endpoint is outside $V(A)$ (as $H$ has no edge in $V(A)$). Therefore, 
the endpoint of $Q$ in $V(A)$ is the endpoint of a path of
class~\ref{t:1} of $\Pi$. Suppose that this path connects
$a_j$ to $v_j\in J$ and the other endpoint of the path $Q$ is $b_j$. As the
inner vector $\mathbf{a}'$ of $\Pi'$ is compatible with
outer vector $\mathbf{b}$, we get that $a'_{j}$ and $b_j$ are
adjacent in $H$. It follows that $P'$ contains a valid path from
$a'_j$ to $b_j$ (note that this path may reenter $V(A)$ several times,
thus we need to use again that $S_0$ and $M$ are the same in both partial
solutions).

We have shown that $\Pi'$ is replacable in $P$, resulting in a
solution $P'$. Thus we have shown that $\R'$ is a
representative set of partial solutions.
\end{proof}

We are now able to present the proof of Lemma~\ref{lem:irrelevantsep}.

\restatelemma{\restateirrelevantsep*}
\begin{proof}
Let $r^*=\max\{r, 10k^2\}$ and let $\smsep:=k\cdot \Rmx$, where $\Rmx$ is the constant in
  Lemma~\ref{lem:repbound}.  We invoke the algorithm of
  Lemma~\ref{lem:repbound} on the separation $(A,B)$.  If it returns
  an induced matching of size $r^*$ or an induced skew biclique on $r^*+r^*$ vertices, then we are done (as $r^*\ge r$). If
  Lemma~\ref{lem:repbound} returns a clique of size $r^*\ge 10k^2$, then we invoke
  Lemma~\ref{lem:cliquereduce}, which either returns $k$-disjoint
  valid paths or an irrelevant terminal; we are done in both
  cases.  Otherwise, let $\R$ be the representative set of size at
  most $\Rmx$ returned by the algorithm of
  Lemma~\ref{lem:repbound}. Each partial solution of $\R$ uses at most $k$ terminals of $V(A)\cap T$ as
  endpoints. Therefore, if we let $T^*$ contain every terminal
  that is an endpoint of a path in one of the partial solutions in $\R$, then we
  have $|T^*|\le k|\R|\le k\cdot \Rmx=\smsep$. The assumption $|V(A)\cap
  T|>\smsep$ implies that there is a $t\in (V(A)\cap T)\setminus T^*$. We
  claim that removing $t$ from the set of terminals does not change
  the solvability of the instance.  Let $P$ be a solution and let
  $\Pi$ be its partial solution at $(A,B)$. If $t$ is not the endpoint
  of path in $P$, then the solution remains a valid even after
  removing $t$ from the set of terminals.  Otherwise, as $\R$ is
  representative, there is a partial solution $\Pi'\in \R$ that is
  replacable into $P$; let $P'$ be the resulting solution. By the
  definition, if $P'$ has a path ending in $V(A)\cap T$, then this
  terminal is endpoint of a path in $\Pi'$ and hence in
  $T^*$. Therefore, $t\not\in T^*$ is not the endpoint of any of the
  paths in $P'$.  This means that $P'$ is a valid solution after
  removing $t$ from the set of terminals and hence $t$ is an
  irrelevant terminal.
\end{proof}

\subsection{Reducing the number of components}
\label{sec:boundcomponents}

With repeated applications of Lemma~\ref{lem:irrelevantsep}, we can
reduce the number of terminals in each component to at most a constant
$\smsep$. The final step of the algorithm is to reduce the number of
components that contain terminals. (We remark that it would be possible to
reduce also the number of components {\em not} having any terminals at all, as their only role is to provide connectivity to $Z$, but we do not need this stronger claim here.)

\newcommand{\tb}{q}
\begin{lemma}[Reducing the number of components of $G-Z$]\label{lem:boundcomp}
  Let $Z\subseteq V(G)$ be a set of vertices disjoint from $T$ such
  that for every component $C$ of $G-Z$, we have $|T\cap V(C)|\le \tb$ and the set $T\cap V(C)$ is
  independent in $H$. If $|T|>100|Z|^4\tb^2$, then we can identify an irrelevant
  terminal in polynomial time.
\end{lemma}
\begin{proof}
 For every ordered pair $(z_{1},z_{2})$ of vertices in $Z$ (possibly with
  $z_{1}=z_{2}$), we mark some of the terminals. We
  proceed the following way for the pair $(z_{1},z_{2})$.  Let
  $\T_{(z_{1},z_{2})}$ contain every ordered pair $(t_1,t_2)$ of terminals with
  the following properties:
\begin{itemize}
\item $t_1$ and $t_2$ are adjacent in $H$.
\item There is a $t_1-z_1$ path whose internal vertices are disjoint from $Z$.
\item There is a $t_2-z_2$ path whose internal vertices are disjoint from $Z$.
\end{itemize}
Clearly, the collection $\T_{(z_{1},z_{2})}$ can be constructed in
polynomial time. Note that by the requirement that $t_1$ and $t_2$ are
adjacent in $H$, we have that $t_1$ and $t_2$ are in different
components of $G-Z$ for every $(t_1,t_2)\in \T_{(z_{1},z_{2})}$.

 Let $b=2|Z|\tb+1$.  First, let us select greedily a maximal
collection of pairs from $\T_{(z_{1},z_{2})}$ such that every
terminal appears in at most one select pair.  If we find $b$ such pairs,
then we mark the (exactly) $2b$ terminals appearing in these pairs
and we are done with processing $(z_{1},z_{2})$. If we do not find
$b$ such pairs, then this means that we can find a set $X$ of at most
$2(b-1)$ terminals such that every pair of $\T_{(z_{1},z_{2})}$
contains a terminal from $X$ (either at the first or second
coordinate). Let us mark every terminal in $X$. Furthermore, for
every $u\in X$, let us mark $b$ terminals $t^*$ such that $(t^*,u)\in
\T_{(z_{1},z_{2})}$ (or all of them if there are less than $b$
such terminals). This completes the description of the marking procedure. We are considering $|Z|^2$ pairs $(z_{1},z_{2})$ and for each pair, we mark at most $\max\{2b, 2(b-1)\cdot (b+1)\}=2(b-1)\cdot (b+1)$ terminals.  Therefore, if there are more than $100|Z|^4\tb^2>|Z|^2\cdot 2(b-1)(b+1)$ terminals, then there is a unmarked terminal. We claim that any unmarked terminal is irrelevant.

Let $t$ be an unmarked terminal and consider a
solution to the instance where $t$ is the endpoint of a
path $P$ of the solution; let $u$ be the other endpoint of $P$. By assumption, $G-Z$ has no valid path and $Z$ is disjoint from $T$, thus path $P$ contains at least one vertex of $Z$. 
Starting at
$v$, let $z_1$ and $z_2$ be the first and last vertices of $P$ in $Z$,
respectively (it is possible that $z_1=z_2$).  Then path $P$ shows
that $(t,u)$ appears in the collection $\T_{(z_1,z_2)}$. Consider
first the case when the marking procedure for $(z_1,z_2)$ found $b$
pairs not sharing any terminals. Observe that the paths of the
solution intersect at most $2|Z|$ components of $G-Z$: each path contains at
least one vertex of $Z$ and these $|Z|$ vertices can break the paths
of the solution into at most $2|Z|$ subpaths. This means that there
are at most $2|Z|\tb$ terminals that are in a component of $G-Z$ intersected by
the solution. Therefore, as we have found $b=2|Z|\tb+1$ pairs, there
is a pair $(t_1,t_2)$ among them such that the components of $t_1$ and
$t_2$ in $G-Z$ are disjoint from the solution. As $(t_1,t_2)\in
\T_{(z_1,z_2)}$, the definition of $\T_{(z_1,z_2)}$ implies that $t_1$
and $t_2$ are adjacent in $H$ (which means that they are in different components of $G-Z$). Furthermore, for $i=1,2$, we can choose a $t_i-z_i$
path $P_i$ whose internal vertices are disjoint from $Z$. This means that the internal vertices of $P_i$ are in the  same component of $G-Z$ as $t_i$, implying that they are disjoint from the solution. We modify the solution:
we replace the $v-z_1$ subpath of $P$ with the $t_1-z_1$ path $P_1$
and the $z_2-u$ subpath of $P$ with the $z_2-t_2$ path $P_2$. This gives 
a valid $t_1-t_2$ path that is disjoint from every other path in the
solution.  Therefore, we have found a solution not involving the
terminal $t$.

Consider now the case when the marking procedure did not find $b$
pairs and hence found a set $X$ of at most $2(b-1)$ terminals. As
$(t,u)\in \T_{(z_1,z_2)}$, either $t$ or $u$ is in $X$. If $t$ is in
$X$, then we marked $t$, thus let us assume that $u$ is in $X$. Then
we marked some terminals $t^*$ for which $(t^*,u)$ is in
$\T_{(z_1,z_1)}$. If $t$ itself was not marked this way, then we
marked $b=2|Z|+1$ such terminals $t^*$. As the solution intersects at
most $2|Z|$ components of $G-Z$ and each component contains at most
$\tb$ terminals, there is a marked terminal $t^*$ whose component is
disjoint from the solution and $(t^*,u)$ is in $\T_{(z_1,z_2)}$. By
the definition of $\T_{(z_1,z_2)}$, this means that $t^*$ and $u$ are
adjacent and the component of $t^*$ is adjacent to $z_1$. Let us
choose a $t^*-z_1$ path $P^*$ whose internal vertices are in the
component of $t^*$ in $G-Z$ (and hence disjoint from the solution). Let us
modify the path $P$ by replacing the $t-z_1$ subpath with the
$t^*-z_1$ subpath $P$. This way, we obtain a solution not involving
the terminal $t$ also in this case, showing that $t$ is indeed
irrelevant.
\end{proof}

We are now ready to prove Lemma~\ref{lem:irrelevantmain}.

\restatelemma{\restateirrelevantmain*}
\begin{proof}
Let $\sm:=100z^4(\smsep)^2$.
Suppose first that a component $C$ of $G-Z$ contains more than $\smsep$ terminals. Then let $(A,B)$ be the separation of $G$ with $V(A)\setminus V(B)=C$
and let us invoke the algorithm of Lemma~\ref{lem:irrelevantsep}. It either returns an irrelevant terminal, a solution with $k$ paths, an induced matching in $H$, or an induced skew biclique in $H$; in all cases, we are done. Assume therefore that every component $C$ of $G-Z$ contains at most $\smsep$ terminals. Then the algorithm of Lemma~\ref{lem:boundcomp} gives an irrelevant terminal.
\end{proof}

\fi

\iffull
\section{Hardness results: matchings}
\label{sec:hardmatching}

In this section, we prove the \classWone-hardness of \maxdisjoint if there is no restiction on the demand pattern.  We are reducing from the following problem:

\defparproblemu{\gridtiling}{For every $1\leq i,j\leq k$, a subset $S_{i,j}\subseteq [n]\times [n]$.}{$k$}{A pair $s_{i,j}\in S_{i,j}$ for every $1\le i,j\le k$ such that
\begin{itemize}[noitemsep]
\item[(i)] for every $1\le i\le k$ and $1\le j < k$, if $s_{i,j}=(x,y)$ and $s_{i,j+1}=(x',y')$, then $x=x'$, and
\item[(ii)] for every $1\le i< k$ and $1\le j \le k$, if $s_{i,j}=(x,y)$ and $s_{i+1,j}=(x',y')$, then $y=y'$.
\end{itemize}
}

\gridtiling is known to be \classWone-hard and reduction from it is a
standard technique for proving \classWone-hardness results for planar
problems (see, e.g.,
\cite{DBLP:conf/icalp/Marx12,DBLP:conf/stacs/MarxP14,DBLP:conf/soda/ChitnisHM14}).

\begin{theorem}\label{th:maxdisjointhard}
\maxdisjoint is \classWone-hard with combined parameters $k$ (the number of paths to be found) and $w$ (where $w$ is the treewidth of $G$), even on planar graphs.\end{theorem}
\begin{proof}
  The proof is by reduction from \gridtiling; let $S_{i,j}\subseteq [n]\times[n] $
  ($1\le i,j\le k$) be the set of pairs in the \gridtiling
  instance. We construct an equivalent instance of \maxdisjoint with parameter
  $k'=4k^2+2k(k-1)$ and treewidth bounded by a function of $k$; this proves the \classWone-hardness of the problem with combined paramters $k$ and $w$. Let us fix an arbitrary bijection $\iota\colon [n]\times [n]\to [n^2]$, e.g., $\iota(x,y)=(x-1)n+y$.

  For each set $S_{i,j}$, we construct a gadget $G_{i,j}$ that is a
  cycle of $4(2n^2+2)$ vertices. The vertices of the cycle are denoted
  by $a_{1,0}$, $b_{1,1,}$, $a_{1,1}$, $b_{1,2}$, $\dots$,
  $b_{1,n^2}$, $a_{1,n^2}$, $c_1$, $a_{2,0}$, $\dots$, $a_{2,n^2}$,
  $c_2$, $a_{3,0}$, $\dots$, $a_{3,n^2}$, $c_3$, $a_{4,0}$, $\dots$,
  $a_{4,n^2}$, $c_4$ (in clockwise order; see
  Figure~\ref{fig:cyclegadget}). For every $1\le i \le k$, $1\le j <
  k$, we introduce a horizontal connector vertex $h_{i,j}$ and make it
  adjacent to vertices $b_{1,1}$, $\dots$, $b_{1,n^2}$ of $G_{i,j}$
  and vertices $b_{3,1}$, $\dots$, $b_{3,n^2}$ of $G_{i,j+1}$.  For
  every $1\le i < k$, $1\le j \le k$, we introduce a vertical
  connector vertex $v_{i,j}$ and make it adjacent to vertices
  $b_{2,1}$, $\dots$, $b_{2,n^2}$ of $G_{i,j}$ and vertices $b_{4,1}$,
  $\dots$, $b_{4,n^2}$ of $G_{i+1,j}$ (see
  Figure~\ref{fig:matrix}). This completes the construction of the
  graph $G$.  It is easy to see that the treewidth of $G$ is bounded
  by a function of $k$: removing the $O(k^2)$ vertices: $h_{i,j}$,
  $v_{i,j}$ results in a graph with treewidth 2 (as every component is
  a cycle), which implies that the treewidth of $G$ is $O(k^2)$. In
  fact, with a bit more effort, one can show that the $G$ has
  treewidth $O(k)$ (details omitted).

\begin{figure}
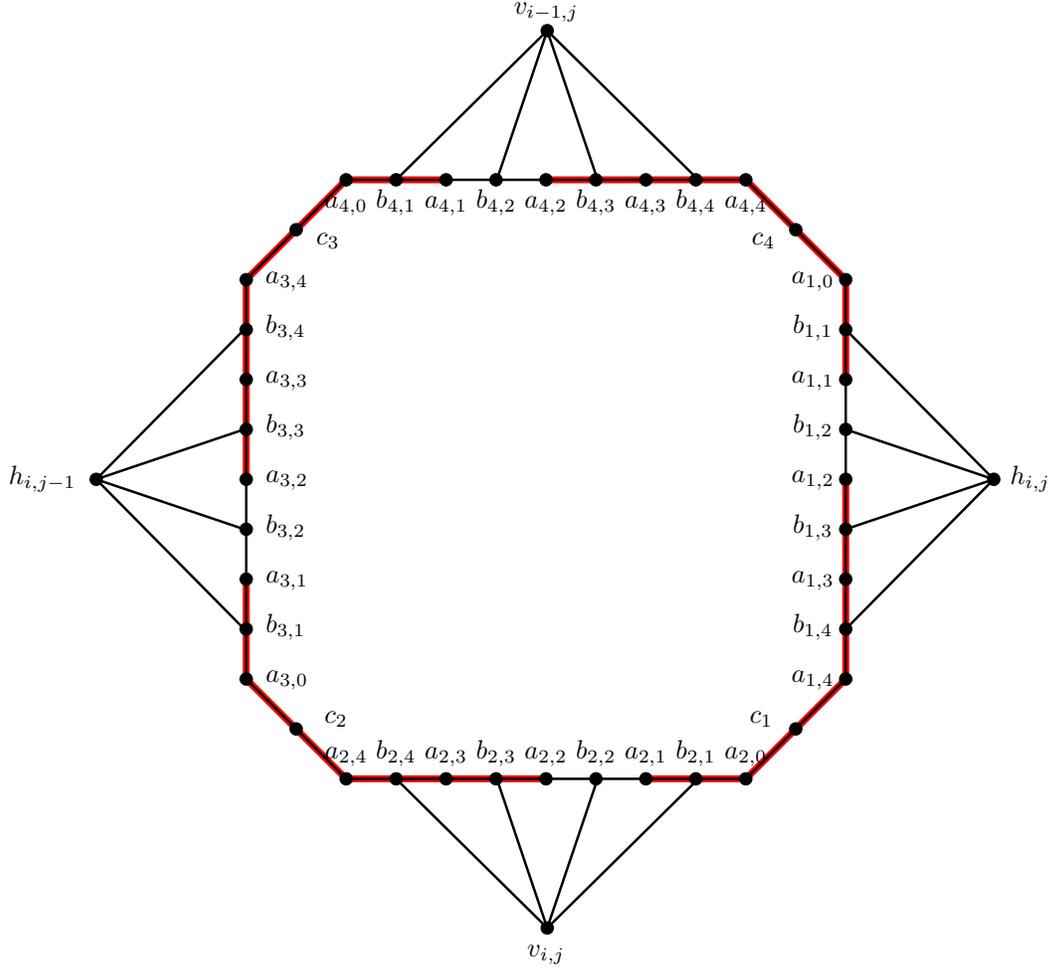

\centering
{\small \svg{0.82\linewidth}{cyclegadget}}
\caption{Proof of Theorem \ref{th:maxdisjointhard}: The gadget $G_{i,j}$ representing the set $S_{i,j}$ with the adjacent two horizontal and two vertical connector vertices. The four red paths show a valid way of realizing 4 disjoint paths in the gadget. }\label{fig:cyclegadget}
\end{figure}

\begin{figure}
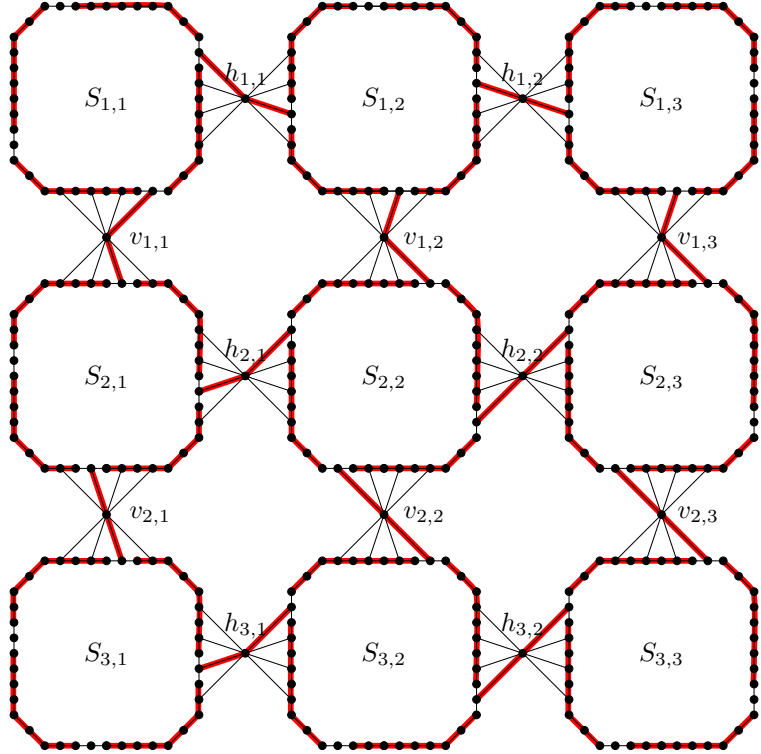

\centering
{\small \svg{0.6\linewidth}{matrix}}
\caption{Proof of Theorem \ref{th:maxdisjointhard}: Connecting the gadgets representing the sets $S_{i,j}$.}\label{fig:matrix}
\end{figure}

The demand pairs of the \maxdisjoint instance are defined the following way.
\begin{itemize}
\item For every $1\le i,j \le k$, $1\le s\le 4$, and every $(x,y)\in S_{i,j}$, we introduce the pair $(a_{s,\iota(x,y)},a_{s+1,\iota(x,y)-1})$ in gadget $G_{i,j}$ (where addition in the first subscript is modulo 4).
\item For every $1\le i \le k$, $1\le j < k$, $(x,y)\in S_{i,j}$, and $(x',y')\in S_{i,j+1}$ with $x=x'$, we introduce the pair consisting of vertex $b_{1,\iota(x,y)}$ of $G_{i,j}$ and vertex $b_{3,\iota(x',y')}$ of $G_{i,j+1}$.
\item For every $1\le i < k$, $1\le j \le k$, $(x,y)\in S_{i,j}$, and $(x',y')\in S_{i,j+1}$ with $y=y'$, we introduce the pair consisting of vertex $b_{2,\iota(x,y)}$ of $G_{i,j}$ and vertex $b_{4,\iota(x',y')}$ of $G_{i+1,j}$.\end{itemize}
This completes the description of the constructed instance of \maxdisjoint. 

Suppose that \gridtiling has a solution $s_{i,j}\in S_{i,j}$, $1\le i,j\le k$. Then we can define $k'$ disjoint paths the following way:
\begin{itemize}
\item For every $1\le i, j \le k$ and $1\le s\le 4$, we select a path in $G_{i,j}$ that goes from $a_{s,\iota(s_{i,j})}$ to $a_{s+1,\iota(s_{i,j})-1}$ clockwise on the cycle.
\item For every $1\le i \le k$ and $1\le j < k$, we select the path that consist of three vertices: vertex $b_{1,\iota(s_{i,j})}$ of $G_{i,j}$, vertex $h_{i,j}$, and vertex $b_{3,\iota(s_{i,j+1})}$ of $G_{i,j+1}$.
\item For every $1\le i < k$ and $1\le j \le k$, we select the path that consist of three vertices: vertex $b_{2,\iota(s_{i,j})}$ of $G_{i,j}$, vertex $v_{i,j}$, and vertex $b_{4,\iota(s_{i+1,j})}$ of $G_{i+1,j}$.
\end{itemize}
It is easy to see that these paths are vertex disjoint and for each path, the endpoints form a pair listed in the \maxdisjoint instance. For example, 
we know that $s_{i,j}$ and $s_{i,j+1}$ have the same first coordinate (by the definition of \gridtiling), hence there is a demand pair consisting of $b_{1,\iota(s_{i,j})}$ of $G_{i,j}$ and $b_{3,\iota(s_{i,j+1})}$ of $G_{i,j+1}$.

For the proof of the reverse direction, suppose that the \maxdisjoint
instance has a solution with $k'$ paths.  Let $S$ be the following set
of vertices: vertices $c_1$, $c_2$, $c_3$, $c_4$ from every gadget
$G_{i,j}$, every horizontal connector $h_{i,j}$, and every vertical
connector $v_{i,j}$. Observe that $S$ has size exactly
$4k+2k(k-1)=k'$ and there is no valid path in $G-S$: no component of
$G-S$ contains the two endpoints of some demand pair. Therefore, each
of the $k'$ paths of the solution has to go through $S$ and hence every
path goes through exactly one vertex of $S$ and each vertex of $S$ is
used by a path of the solution.

Consider the path of the solution that goes through vertex $c_1$ of $G_{i,j}$. As it does not go through any other vertex of $S$ (in particular, it does not go thorough $c_2$ and $c_4$ of $G_{i,j}$, the horizontal connector $h_{i,j}$, and the vertical connector $v_{i,j}$), the endpoints of this path have to be $a_{1,t}$ and $a_{2,t-1}$ of $G_{i,j}$ for some $1\le t \le n^2$. Similarly, for every $1\le s \le 4$, the solution contains a path going from $a_{s,t_s}$ to $a_{s+1,t_s-1}$ of $G_{i,j}$ in clockwise direction on the cycle. As these paths are vertex disjoint, we have, for example, $t_1\le t_{2}$, as otherwise the two paths would both contain the vertex $a_{2,t_1}$. Therefore, we get the cycle of equalities $t_1\le t_2 \le t_3 \le t_4 \le t_1$, implying that all these four numbers are equal. This means that there is a $1\le t_{i,j} \le n^2$ such that the solution selects the four paths with endpoints $(a_{1,t_{i,j}},a_{2,t_{i,j}-1})$,  $(a_{2,t_{i,j}},a_{3,t_{i,j}-1})$, $(a_{3,t_{i,j}},a_{4,t_{i,j}-1})$, and $(a_{4,t_{i,j}},a_{1,t_{i,j}-1})$ in $G_{i,j}$.
The existence of these demands pair imply that $t_{i,j}=\iota(s_{i,j})$ for some $s_{i,j}\in S_{i,j}$. We claim that these values $s_{i,j}$ define a solution of the \gridtiling instance.

Observe that the 4 paths of the solution in $G_{i,j}$ leave only the 4
vertices $b_{1,t_{i,j}}$, $b_{2,t_{i,j}}$, $b_{3,t_{i,j}}$, and
$b_{4,t_{i,j}}$ unoccupied on the cycle of $G_{i,j}$. Therefore, the
path of the solution that goes through $h_{i,j}$ consists of 
vertex $b_{1,t_{i,j}}$ of $G_{i,j}$, vertex $h_{i,j}$, and vertex $b_{3,t_{i,j+1}}$ of $G_{i,j+1}$. The fact that vertex $b_{1,t_{i,j}}$ of $G_{i,j}$ and vertex $b_{3,t_{i,j+1}}$ of $G_{i,j+1}$ form a demand pair in the \maxdisjoint instance implies that the first coordinate of $s_{i,j}$ and the first coordinate of $s_{i,j+1}$ are the same. In a similar way, by looking at the path of the solution going through vertex $v_{i,j}$, we can deduce that the second coordinate of $s_{i,j}$ and the second coordinate of $s_{i+1,j}$ are the same. Thus the $s_{i,j}$'s indeed form a solution of the \gridtiling instance.
\end{proof}

\section{Hardness results: skew bicliques}
\label{sec:hardskew}

In the section, we prove the \classWone-hardness of the following specific disjoint path problem.

\defparproblemu{\skewdisjoint}
{
A graph $G$ with terminals $s_1$, $\dots$, $s_n$, $t_1$, $\dots$, $t_n$, an integer $k$.
}{}
{
A set of $k$ pairwise vertex-disjoint paths such that each path connects some $s_i$ and some $t_j$ with $i\le j$.
}

The \classWone-hardness of \skewdisjoint clearly implies that \maxhhdisjoint is \classWone-hard if $\cH$ contains every skew biclique.
 
For technical reasons, it will be convenient to define \skewdisjoint in a slightly different way:

\defparproblemu{\skewdisjointx}
{
A graph $G$, a set $T\subseteq V(G)$, and a labeling function $\mu:T\to \mathbb{Z}\setminus \{0\}$.
}{}
{
A set of $k$ pairwise vertex-disjoint paths such that if $u$ and $v$ are the endpoints of a path, then we have
\begin{itemize}[noitemsep]
\item $\mu(u)\mu(v)<0$ and
\item $\mu(u)+\mu(v)\le 0$.
\end{itemize}
}

Note that in \skewdisjointx, the labeling $\mu$ is not necessarily injective, i.e., two terminals can have the same label.
It is easy to see that the two versions of \skewdisjoint are
equivalent. 
\begin{lemma}\label{lem:skewequiv}
The are parameterized reductions between \skewdisjoint and \skewdisjointx.
\end{lemma}
\begin{proof}
  To transform an instance of \skewdisjoint to \skewdisjointx, we
  first modify the instance so that every vertex is used as at most
  one $s_i$ or $t_i$: this can be achieved by attaching sufficiently
  many degree-1 vertices to each vertex and then replacing each $s_i$
  and $t_i$ with an adjacent degree-1 vertex that was not used before.
  Then we define $T$ to be the set of terminals and set the labels as
  $\mu(s_i)=i$ and $\mu(t_i)=-i$ for every $1\le i \le n$. The sign of
  the labels ensure that every valid path connects some $s_i$ with
  some $t_j$, and the condition $\mu(s_i)+\mu(t_j)\le 0$ ensures $i\le
  j$.

  For the other direction, when transforming an instance of
  \skewdisjointx to \skewdisjoint, we first ensure that the labeling
  is injective.  Let $T=\{v_1,\dots,v_{|T|}\}$, ordered by increasing
  order of labels. We define $\mu'(v_i)=2\mu(v_i)|T|-(i-1)$. Note that
  $\mu(v_i)$ and $\mu'(v_i)$ have the same sign and
  $\mu(v_i)+\mu(v_j)\le 0$ if and only if $\mu'(v_i)+\mu'(v_j)\le 0$,
  thus replacing $\mu$ with $\mu'$ does not change the problem.

Next we want to ensure that $|\mu'(v_i)|\le 2|T|$ for every $v_i$. If
this is not true, then there has to be an $1\le x \le 2|T|$ such that
neighter $x$ nor $-x$ appears in the image of $\mu'$. Then let us
decrease the value of $\mu'(v_i)$ by one if it is greater than $x$ and
let us increase the value of $\mu'(v_i)$ by one if it is less than
$-x$. Again, this transformation does not change the instance. Let us
repeat this step until we get a labelling $\mu''$ with $m:=\max
|\mu'(v_i)|\le 2|T|$.

Next we ensure that the image of the labeling function is exactly $[-m,-1]\cup [1,m]$: if there is an integer $x$ in this range that does not appear in the image of $\mu''$, the let us introduce a new isolated vertex $v$ and let us define $\mu''(v)=x$. If the image of $\mu''$ is $[-m,-1]\cup [1,m]$, then we interpret the problem as a \skewdisjoint instance by defining $s_i$ (resp., $t_i$) to be the unique $v\in
T$ with $\mu''(v)=i$ (resp., $-i$). It is clear that the resulting \skewdisjoint instance is equivalent to the original \skewdisjointx instance.
\end{proof}

As in Section~\ref{sec:hardmatching}, \classWone-hardness is proved by reduction from \gridtiling. To reduce \gridtiling to \skewdisjointx, we construct certain
gadgets. Formally, a {\em gadget} is a graph $G$ with a set $B\subseteq
V(G)$ of boundary vertices, a $T\subseteq V(G)$ of terminals, and an
injective function $\mu\colon T\to \mathbb{Z}\setminus \{0\}$. We often
describe the boundary vertices as an ordered tuple $(b_1,\dots,b_r)$
of vertices. We assume that $B\cap T=\emptyset$, that is, the boundary
vertices are not labeled. Given two gadgets, we can join them by
identifying some of their boundary vertices; the set of terminals
becomes the union of the two sets and the function $\mu$ is defined
the obvious way on the union. 

A {\em partial solution} in a gadget is a set of pairwise vertex-disjoint paths, where every path is either
\begin{itemize}[noitemsep]
\item a {\em complete path} connecting two vertices $u,v\in T$ and satisfying $\mu(u)\mu(v)<0$ and $\mu(u)+\mu(v)\le 0$, or
\item a {\em partial path} connecting a vertex $u\in T$ and a vertex $v\in B$. 
\end{itemize}
If the boundary of $G$ is $(b_1,\dots, b_r)$, then we say that a
partial solution represents the tuple $(x_1,\dots, x_r)$ if for every
$1\le i\le r$, the partial solution contains a partial path with
endpoints $b_i$ and $v_i\in T$ with $\mu(v_i)=x_i$. Note that this implicitly 
implies that the partial solution contains exactly $r$ partial
paths.

Our reduction is based on the existence of gadgets defined by the following lemma.

\begin{restatable}{\relemma}{restatemaingadget}\label{lem:maingadget}
Let $n$ and $B> 8n^2$ be integers. Given a subset $S\subseteq [n]\times [n]$, one can construct in polynomial time a {\em positive gadget} $G$ such that the following holds:
\begin{enumerate}[noitemsep]
\item $G$ is a planar graph of constant treewidth and the boundary vertices $(b_1,\dots, b_8)$ appear in this order on a single face.
\item For every $(x,y)\in S$, gadget $G$ has a partial solution containing $6$ complete paths and $8$ partial paths representing
\[
t_{(x,y)}:=(B+x,B-x,
B+y,B-y,
B+x,B-x,
B+y,B-y).\]
\item Every partial solution of $G$ contains at most $6$ complete paths. 
\item If a partial solution of $G$ contains exactly $6$ complete paths and 8 partial paths, then it represents the tuple 
$t_{(x,y)}$ for some $(x,y)\in S$.
\end{enumerate}
The definition of the {\em negative gadget} is the same except that we require  $B <-n$.
\end{restatable}
As the boundary vertices of the gadget given by
Lemma~\ref{lem:maingadget} are on a single face, we may embed the gadget
in such a way that the boundary vertices appear on the infinite face. Intuitively, we call $b_1$
and $b_2$ as the right boundary vertices, $b_3$ and $b_4$ as the bottom
boundary vertices, $b_5$ and $b_6$ as the left boundary vertices,
and $b_7$ and $b_8$ as the top boundary vertices. That is, $b_1$,
$\dots$, $b_8$ appear in clockwise order around the gadget (see
Figure~\ref{fig:posneg}).
\begin{figure}
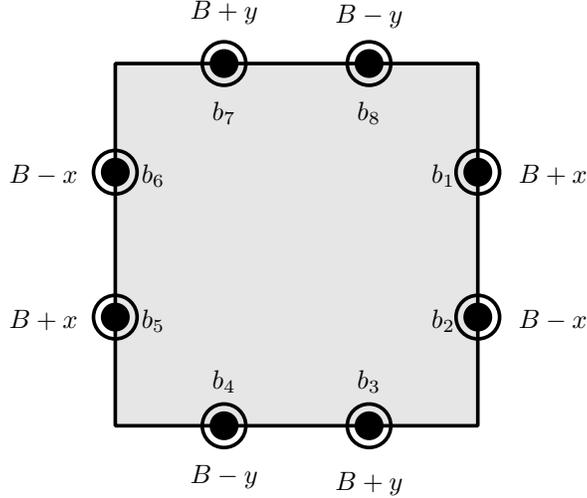

\begin{center}
{\small \svg{0.55\linewidth}{posneg}}
\caption{A gadget constructed by Lemma~\ref{lem:maingadget}.}\label{fig:posneg}
\end{center}
\end{figure}
We prove Lemma~\ref{lem:maingadget} in Section~\ref{sec:constructing-gadgets}. 
Assuming the existence of such gadgets, we prove the \classWone-hardness of \skewdisjointx.
\begin{theorem}\label{th:skewhard}
\skewdisjointx is \classWone-hard with combined parameters $k$ (the number of paths to be found) and $w$ (where $w$ is the treewidth of $G$), even on planar graphs.
\end{theorem}
\begin{proof}
  The proof is by reduction from \gridtiling; let $S_{i,j}\subseteq [n]\times [n] $
  ($1\le i,j\le k$) be the set of pairs in the \gridtiling
  instance. We construct an instance of \skewdisjoint with parameter
  $k'=O(k^2)$ and treewidth bounded by a function of $k$.

For each set $S_{i,j}$, we use
  Lemma~\ref{lem:maingadget} to construct a gadget $G_{i,j}$ corresponding to the set $S_{i,j}$ as follows.
Let $Z=10n^2$.  If $i+j$ is even, then $G_{i,j}$ is a positive gadget with
  parameters $n$ and $B=Z$.
  If $i+j$ is odd, then $G_{i,j}$ is a negative gadget with parameters
  $n$ and $B=-Z$. After constructing these $k^2$ gadgets, we join the gadgets the following way: for $1\le i \le
  k$ and $1\le j<k$, we identify the right boundary vertices of
  $G_{i,j}$ with the left boundary vertices of $G_{i,j+1}$; and for
  $1\le i < k$ and $1\le j \le k$, we identify the bottom boundary
  vertices of $G_{i,j}$ with the top boundary vertices of $G_{i+1,j}$
  (see Figure~\ref{fig:main}). This way, the $8k^2$ boundary vertices of the $k^2$ gadgets
  are identified into a set $X$ of $4k(k+1)$ vertices. 
  The set $X$ contains $8k$ vertices that came from a single gadget,
  that is, they were not identified with other boundary vertices. For
  example, the top boundary vertices of the gadgets $G_{1,1}$,
  $\dots$, $G_{1,k}$ are such vertices. Let $X_1\subseteq X$ be this set of $8k$ vertices and let $X_2=X\setminus X_1$. 
  We label with $-Z-n$ each vertex $v\in X_1$ that appears in a
  positive gadget and we label with $Z-n$ each vertex $v\in X_1$
  that appears in a negative gadget.  This completes the description
  of the constructed graph $G$. It is easy to observe that the treewidth of the
  $G$ is $O(k^2)$: after removing the set $X$ (which
  has size $O(k^2)$), the instance falls apart into components whose
  treewidth is bounded by a constant (property 1 of
  Lemma~\ref{lem:maingadget}). It is possible to prove that treewidth is actually $O(k)$ (details omitted).
\begin{figure}
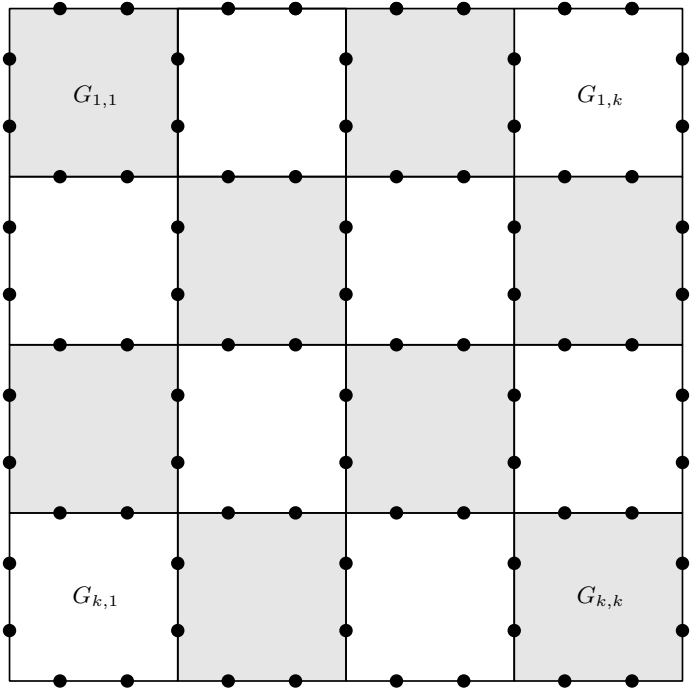

\begin{center}
{\footnotesize \svg{0.55\linewidth}{main1}}
\caption{Connecting the gadgets in the proof of Theorem~\ref{th:skewhard}. The shaded boxes are positive gadgets, then white boxes are negative gadgets.}\label{fig:main}
\end{center}
\end{figure}

Set $k'=4k(k+1)+6k^2$. We claim that the \skewdisjointx instance has a
solution with $k'$ paths if and only if the \gridtiling instance has a
solution. Suppose first that $s_{i,j}\in S_{i,j}$ ($1\le i,j\le k$) is
a solution of \gridtiling; we construct a solution for \skewdisjointx
with $k'$ paths as follows (see Figure~\ref{fig:main2}). For each gadget,
property 3 of Lemma~\ref{lem:maingadget} gives a partial solution with
6 complete paths and 8 partial paths going to the boundary
vertices, representing the 8-tuple $t_{s_{i,j}}$. Suppose that $s_{i,j}=(x,y_1)$ and $s_{i,j+1}=(x,y_2)$;
recall that, by the definition of \gridtiling, they have to agree on
the first coordinate. Suppose that $i+j$ is even. Then the right
boundary vertex $b_1$ of positive gadget $G_{i,j}$ was identified with
the left boundary vertex $b_6$ of the negative gadget $G_{i,j+1}$; let
$v\in X$ be this identified vertex. Therefore, in the partial solution
of $G_{i,j}$, vertex $v$ is connected to vertex with label $Z+x$,
while the partial solution of $G_{i,j+1}$ connects $x$ to a
vertex with label $-Z-x$. Thus the two partial paths create a valid
path. Similarly, we can verify in all other cases that whenever two
boundary vertices were identified, the two partial paths of the two
gadgets together form a valid path. Finally, let $v\in X_1$ be one of
the $8k$ boundary vertices that are contained only in a single
gadget. If $v$ appears in a positive gadget $G_{i,j}$ and a partial
path connects $x$ to vertex labeled $Z+ z$ for some $-n \le z
\le n$, then this partial path is actually a valid path, as $v$ was
labeled $-Z-n$ in the construction of the instance.  Similarly, if
$v$ is in a negative gadget, then $v$ has label $Z-n$, making the
partial path a valid path. Therefore, we get 6 paths in each of the
$k^2$ gadgets and a separate path going through each of the $4k(k+1)$
vertices of $X$, giving $k'=4k(k+1)+6k^2$ paths in total, as required.
\begin{figure}
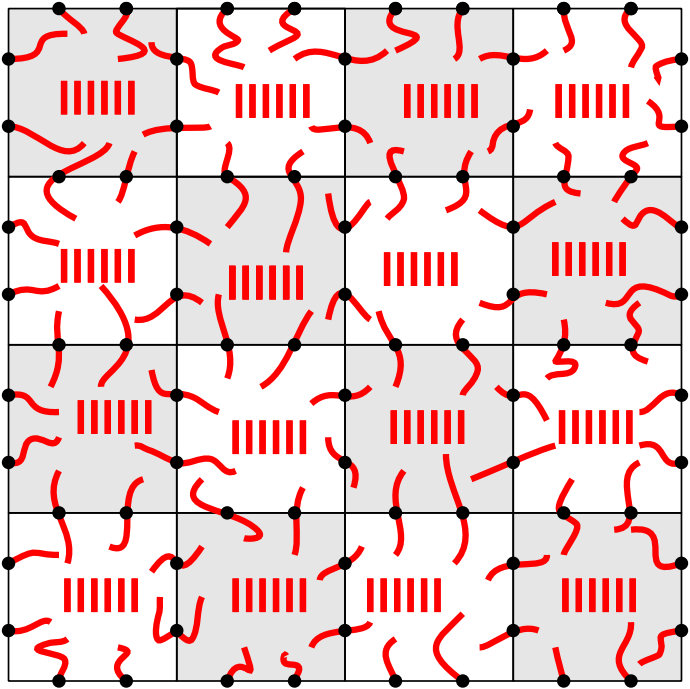

\begin{center}
{\footnotesize \svg{0.55\linewidth}{main2}}
\caption{A solution with $k'=4k(k+1)+6k^2$ paths. There are 6 paths in each gadget and one path going through each boundary vertex.}\label{fig:main2}
\end{center}
\end{figure}

For the reverse direction, consider a solution consisting of $k'$
paths. By property 4 of Lemma~\ref{lem:maingadget}, at most 6 paths
can be fully contained in each of the $k^2$ gadgets. Additionally, at most
$|X|=4k(k+1)$ paths can go through $X$. Therefore, having $k'$ paths
is only possible if each vertex of $X$ is used by a separate path and
exactly 6 paths are fully contained in each gadget. This means that
for every $G_{i,j}$, the solution induces a partial solution of
$G_{i,j}$ with 6 complete paths and 8 partial paths. By property 5 of
Lemma~\ref{lem:maingadget}, this partial solution has to represent a
tuple $t_{(x,y)}$ for some $(x,y)\in S_{i,j}$; let us define $s_{i,j}$
to be this pair $(x,y)\in S_{i,j}$. We claim that these $s_{i,j}$'s
form a solution of \gridtiling. Consider first the case when $i+j$ is
even, we have $s_{i,j}=(x,y)$, $s_{i,j+1}=(x',y')$, and suppose for contradiction that $x>x'$. The
right boundary vertex $b_1$ of positive gadget $G_{i,j}$ was
identified with the left boundary vertex $b_6$ of negative gadget
$G_{i,j+1}$; let $v$ be this identified vertex. We know that
$G_{i,j}$ contains a partial path connecting $v$ to a vertex
labeled $Z + x$ and $G_{i,j}$ contains a partial path connecting
$x_\alpha$ to a vertex labeled $-Z -x'$. Now $x>x'$ implies that $(Z+x)+(-Z-x')>0$ and hence these
two partial paths together do not create valid path, a
contradiction. Suppose now $x<x'$ and let $v$ be the vertex arising
from the identification of the right boundary vertex $b_2$ of
$G_{i,j}$ with the left boundary vertex $b_5$ of $G_{i,j+1}$. Now the
two endpoints of the path going through $x_\alpha$ are labeled $Z-x $
(in $G_{i,j}$) and $-Z +x'$ (in $G_{i,j+1}$), hence $x<x'$ gives a contradiction
again. The situation is similar if $i+j$ is odd, i.e., $G_{i,j}$ is a
negative gadget. Finally, in a similar way, we can show that if
$s_{i,j}=(x,y)$ and $s_{i+1,j}=(x',y')$, then $y=y'$ has to hold by looking at the
paths going through the vertices arising from the identification of
the bottom boundary vertices of $G_{i,j}$ and the top boundary vertices
of $G_{i+1,j}$.
\end{proof}
\subsection{Constructing the gadgets}
\label{sec:constructing-gadgets}

The first step in the consturction of the gadgets required by Lemma~\ref{lem:maingadget} is a selector gadget that has $m$ possible states.
\begin{lemma}\label{lem:selectorgadget}
Given a positive integer $m$, one can construct in polynomial time a gadget $G$ such that the following holds:
\begin{enumerate}[noitemsep]
\item $G$ is an embedded planar graph of constant treewidth with the boundary vertices $(b^+,b^-)$ appearing on a single face.
\item For every $1\le i \le m$, gadget $G$ has a partial solution containing 4 complete paths and two partial paths representing $(6m+i,-m-i)$.
\item Every partial solution of $G$ contains at most 4 complete paths.
\item If a partial solution of $G$ contains exactly 4 complete paths and represents the pair $(x,y)$, then $x=6m+i^+$ and $y=-m-i^-$ for some $1\le i^-\le i^+ \le m$.
\item All the labels in $G$ are in the range $[-5m,7m]$.
\end{enumerate}
\end{lemma}
\begin{proof}
The gadget is demonstrated in Figure~\ref{fig:selector}. Properties 1 and 5 are obvious by inspection. For property 2, consider the following set of complete paths and partial paths:
\begin{itemize}[noitemsep]
\item a partial path $P^+$ connecting $b^+$ and $6m+i$,
\item a partial path $P^-$ connecting $b^-$ and $-m-i$,
\item a complete path $P_1$ connecting $3m-1$ and $-3m+1$,
\item a complete path $P_2$ connecting $4m-i$ and $-4m+i$,
\item a complete path $P_3$ connecting $5m-i$ and $-5m+i$, and
\item a complete path $P_4$ connecting $5m+1$ and $-5m-1$.

\end{itemize}
As shown in Figure~\ref{fig:selector}, these endpoints can be connected by vertex-disjoint paths.
\begin{figure}
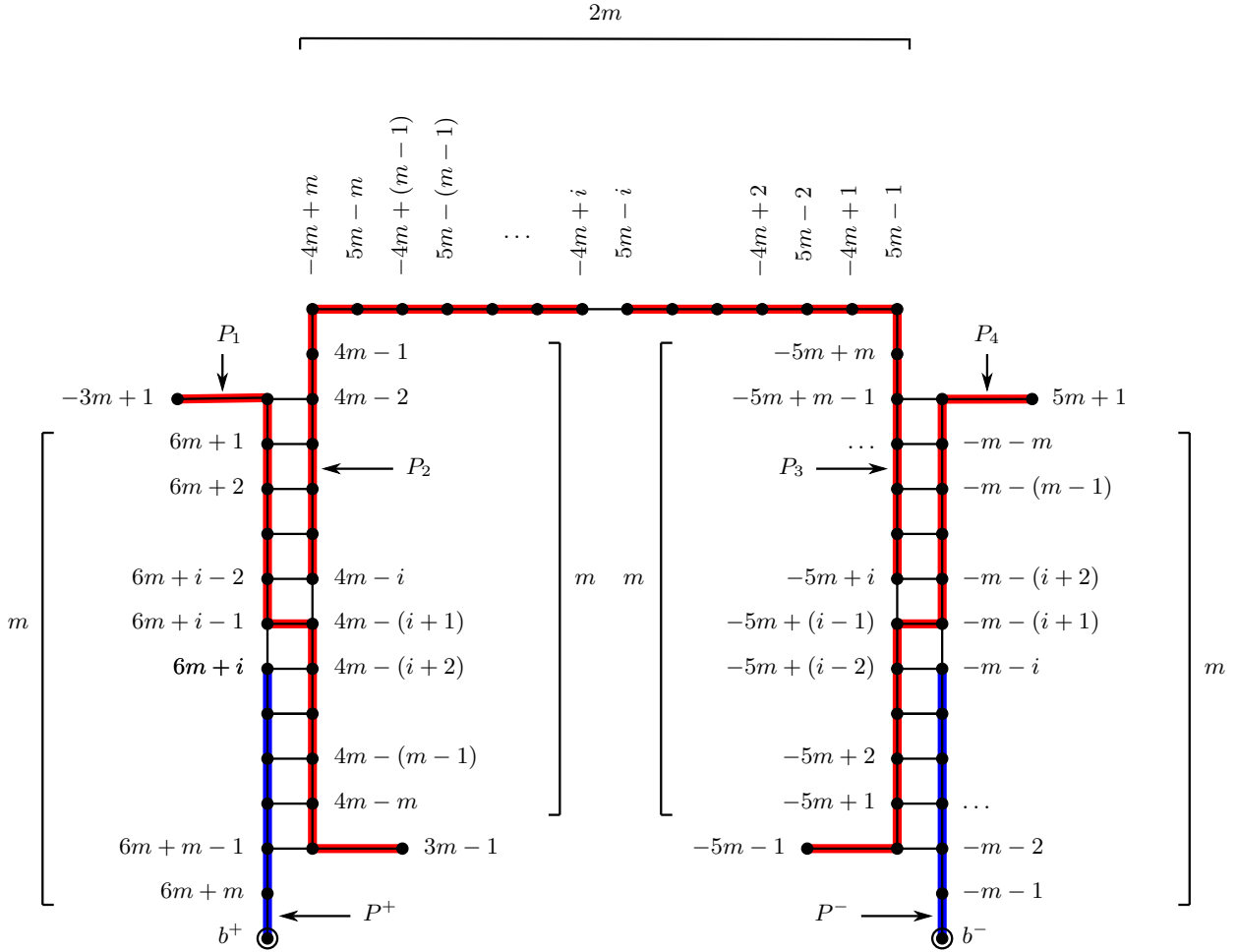

\begin{center}
{\footnotesize \svg{\linewidth}{selector}}
\caption{The selector gadget and a partial solution representing $(6m+i,-m-i)$. The red path are the complete paths, the blue paths are the partial paths.}\label{fig:selector}
\end{center}
\end{figure}

For property 3, observe that if we remove the 4 vertices labeled
$-3m+1$, $4m-1$, $-5m+m$, $5m+1$, then no complete path can be
created in the remaining components of the gadget. This shows that there are
at most 4 complete paths in any partial solution.

For property 4, suppose that
there are exactly 4 complete paths and 2 partial paths in a partial
solution. Then one of the complete paths, call it $P_1$, has to go through the vertex
labeled $-3m+1$; in fact, as this vertex has degree 1, it is the
endpoint of $P_1$. Observe that $P_1$ cannot go through
$4m-1$ (otherwise the number of complete paths is at most
3). Therefore, vertex $3m-1$ is the only vertex with positive label at
most $3m-1$ that is reachable from $-3m+1$, and hence it has to be the
other endpoint of $P_1$.

Let $P^+$ be the partial path with one endpoint in $b^+$. The label of
the other endpoint is in $[6m+1,7m]$, otherwise $P^+$ would
separate $-3m+1$ and $3m-1$, the endpoints of $P_1$. Suppose
therefore that $6m+i^+$ is the other endpoint of $P^+$.  Consider now
the complete path $P_2$ going through vertex $4m-1$. As $P_2$ cannot
go through vertex $-5m+m$, the negative endpoint has label from
$[-4m+1,-4m+m]$ and hence the positive endpoint has a label from
$[4m-m,4m-1]$. Suppose that the label of the positive endpoint is
$4m-j$. We claim that $j\le i^+$. Otherwise (i.e., when $j\ge
i^++1$), consider vertices $6m+i^+$ and $4m-(i^++1)\ge 4m-j$;
note that second vertex is ``one step to the right and above'' to the
first in Figure~\ref{fig:selector}. These two vertices separate the
endpoints of $P_1$.
Now all three of the paths $P^+$, $P_1$, $P_2$ contain at
least one of these two vertices, contradicting that the paths are
disjoint. Thus $j\le i^+$ holds, and the negative endpoint of $P_2$
has label at most $-4m+j \le -4m+i^+$.

Similar arguments show that
\begin{itemize}[noitemsep]
\item there is a complete path $P_4$ connecting $5m+1$ and $-5m-1$,
\item there is a partial path $P^-$ connecting $b^-$ and $-m-i^-$ for some $1\le i^- \le m$, and
\item there is a partial path $P_3$ going through $-5m+m$ whose positive endpoint is at most $5m-i^-$.
\end{itemize}
Summarizing, the negative endpoint of $P_2$ has to be at or to the right of $-4m+i^+$ and the positive endpoint of $P_3$ has to be at or to the left of $5m-i^-$. As $P_2$ and $P_3$ are disjoint, this is only possible if $i^+\ge i^-$, what we had to show.
\end{proof}

The following lemma is a generic construction of a gadget that has 8 outputs and can represent only a prescribed set of 8-tuples on these outputs.
\begin{lemma}\label{lem:generalgadget}
  Let $m$ be a positive integers. Let $t_1,\dots,t_m\in \mathbb{Z}^8$
  be a sequence of 8-tuples where every coordinate is an integer
  greater than $7m$.  Then one can construct in polynomial time a
  gadget $G$ such that the following holds:
\begin{enumerate}[noitemsep]
\item $G$ is an embedded planar graph of constant treewidth and the boundary vertices $(b_1,\dots, b_8)$ appear in this order around a single face.
\item For every $1\le i \le m$, gadget $G$ has partial solution containing $6$ complete paths and $8$ partial paths representing $t_i$.
\item Every partial solution of $G$ contains at most $6$ complete paths. 
\item If a partial solution of $G$ contains exactly $6$ complete paths and 8 partial paths, then it represents $t_i$ for some $1\le i \le m$.
\item Every label appearing on a vertex of $G$ is either in $[-7m,7m]$ or one of the coordinates of some $t_i$. 
\end{enumerate}
\end{lemma}
\begin{proof}
Let $t_i=(t_{i,1},\dots,t_{i,8})$ for every $1\le i \le m$.
  The construction of the gadget starts with an $10\times 10m$ grid
  (see Figure~\ref{fig:gadget}).  Every vertex of the top row is a
  terminal and they are labeled the following way. The columns are
  divided into $m$ blocks of 10 columns each. In block $i$ (where $1\le i \le m$),
  the first vertex is labeled $-6m-i$, the last vertex is labeled
  $m+i$, and the 8 vertices in between are labeled using the components
  of $t_i$, that is, by $t_{i,1}$, $\dots$, $t_{i,8}$ (see
  Figure~\ref{fig:gadget}). 
\begin{figure}
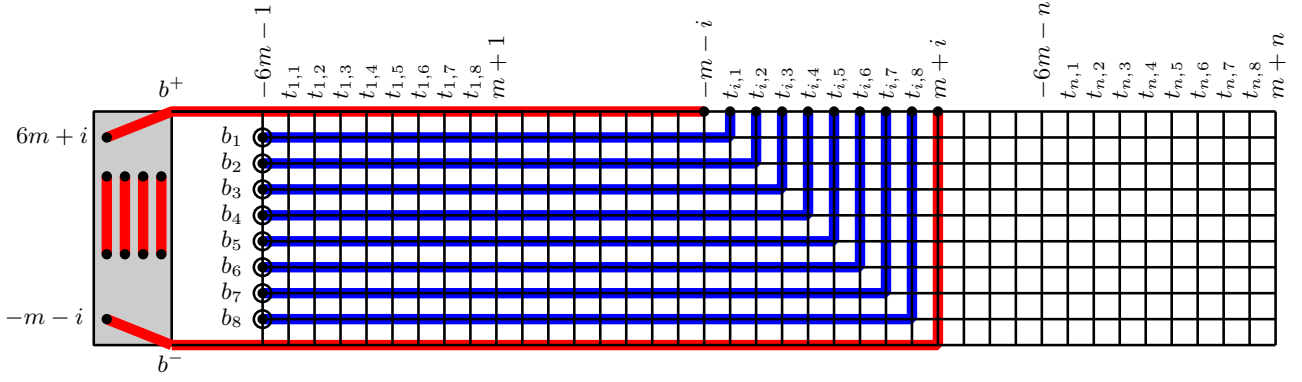

\begin{center}
{\footnotesize \svg{\linewidth}{gadget}}
\caption{Gadget construction of Lemma~\ref{lem:generalgadget}. The figure shows a partial solution representing $(t_{i,1},\dots, t_{i,8}$ and having 6 complete paths (red) and 8 partial paths (blue). Note that 4 of the complete paths are in the selector gadget (shaded box). }\label{fig:gadget}
\end{center}
\end{figure}

Using Lemma~\ref{lem:selectorgadget}, we construct a selector gadget with parameter $m$ and connect boundary vertex $b^+$ of the
selector with the top left vertex of the grid and boundary vertex
$b^-$ of the selector with the bottom left vertex of the grid. The
remaining 8 vertices of the leftmost column are the boundary vertices
of the gadget. 

Property 1 and 5 are clear from the way the gadget is defined (note that the selector gadget uses labels only in the range $[-5m,7m]$.
For property 2, we construct a partial solution the following way. By property 4 of Lemma~\ref{lem:selectorgadget}, the selector gadget has a partial solution with 4 complete paths and two partial paths, one connecting $b^+$ and $6m+i$, the other connecting $b^-$ and $-m-i$. We extend these partial paths by connecting $b^+$ with $-6m-i$ in the top row and $b^-$ with $m+i$ in the top row (see Figure~\ref{fig:gadget}). Then we can connect the terminals between $-6m-i$ and $m+i$ in the top row, that is, the terminals $t_{i,1}$, $\dots$, $t_{i,8}$ to the boundary vertices $b_1$, $\dots$, $b_8$.

To see property 3, observe that no complete path can have both of its endpoints on the top row (here we use that every $t_{i,j}$ is at least $7m+1$). Therefore, every complete path is either inside the selector gadget or uses a boundary vertex of the selector gadget. Property 4 of Lemma~\ref{lem:selectorgadget} implies that there are at most 4 complete paths inside the selector gadget, thus there can be at most 6 complete paths in the gadget we are constructing.

For property 4, consider a partial solution with 6 complete paths and
8 partial paths. As we have seen in the previous paragraph, this is
possible only if there are 4 complete paths completely contained in
the selector gadget and there are two complete paths each connecting a
vertex of the selector gadget to the top row. Let $P^+$ (resp., $P^-$) be the complete path connecting a vertex of the selector gadget to the top row via boundary vertex $b^+$ (resp., $b^-$) of the selector gadget. By property 4 of the selector gadget, we may assume that the positive endpoint of $P^+$ is at a vertex labeled $6m+i^+$ and the negative endpoint of $P^-$ is at a vertex labeled $-6m-i^-$ for some $1\le i^-\le i^+ \le m$. This means that the negative endpoint of $P^+$ is at or to the right of the vertex labeled $-6m-i^+$ on the top row, and the positive endpoint of $P^-$ is at or to the left of the vertex labeled $m+i^-$ on the top row (here we use again that every $t_{i,j}$ is at least $7m+1$). By planarity, this is only possible if $i^+=i^-=i$ and the two endpoints are exactly $-6m-i$ and $m+i$. Again by planarity, this implies that the 8 partial paths connect the boundary vertices $b_1$, $\dots$, $b_8$ to the vertices between 
$-6m-i$ and $m+i$, that is, to the vertices $t_{i,1}$, $\dots$, $t_{i,8}$ and exactly in this order. In other words, the partial solution represents the tuple $t_i$, what we had to show.
\end{proof}

It is quite straightforward to constuct the positive gadget of
Lemma~\ref{lem:maingadget} using Lemma~\ref{lem:generalgadget}. Then
we argue that the negative gadget can be obtained from the positive
gadget by simple transformations.

\restatelemma{\restatemaingadget*}
\begin{proof}
  We construct the positive gadgets as follows. Let $m=|S|\le n^2$ and
  let $t_1$, $\dots$, $t_m$ be an ordering of the tuples $t_{(x,y)}$
  for every $(x,y)\in S$ (as defined in the statement of the lemma).
  As $1\le x,y\le n$, the condition $B>8n^2\ge 8m$ 
  implies that the integers appearing in these tuples are greater than $7m$. Therefore, we can invoke
  Lemma~\ref{lem:generalgadget} to construct the required
  gadget. Properties 1--4 of Lemma~\ref{lem:generalgadget} imply that
  Properties 1--4 of Lemma~\ref{lem:maingadget} are satisfied.

  For the construction of the negative gadget, let us set $\Delta=-B+8n^2+1$ and $B^+:=B+\Delta=8n^2+1$.  Let us construct
  as above the positive gadget $G$ for $n$, $S$, and $B^+$; let $\mu^+$ be the
  labeling of the gadget.  Note that every negative label in $\mu^+$
  is at least $-7n^2\le -7m$ (by Property 5 of Lemma~\ref{lem:generalgadget})
  and every positive label is at most $B^++n=8n^2+n+1$. We obtain the negative
  gadget by defining a new labeling $\mu$ the following way: if
  $\mu^+(v)$ is negative, then let $\mu(v)=\mu^+(v)+\Delta$ (which is a
  positive number); if $\mu^+(v)$ is positive, then let
  $\mu(v)=\mu^+(v)-\Delta$ (which is at most $8n^2+n+1-\Delta=B+n<0$ by the assumption $B<-n$).  We claim that any
  path $P$ is a valid path in the labeling $\mu^+$ if and only if it
  is a valid path in the labeling $\mu$. Indeed, for any two vertices
  $u$ and $v$, the sign of $\mu(u)\mu(v)$ is the same as the sign of
  $\mu^+(u)\mu^+(v)$, and if this sign is negative, then
  $\mu(u)+\mu(v)=\mu^+(u)+\mu^+(v)$ holds (as one of the two labels
  were increased by $\Delta$ and the other was decreased by $\Delta$). Therefore,
  the two gadgets have the same set of valid paths and hence the same
  properties. If a partial path has an endpoint labeled with (the
  positive number) $B^++x$ in the positive gadget, then this
  translates to a partial paths with endpoint labeled with $B+x$ in
  the negative gadget.
\end{proof}


\fi

\bibliographystyle{abbrv}
\bibliography{fptmaxdisjoint}
\end{document}